\documentclass[11pt]{article}
\usepackage[margin=1in]{geometry}
\usepackage{amsmath,amsfonts,amssymb,amsthm}
\usepackage{graphicx}
\usepackage{enumerate}
\usepackage{bbm}
\usepackage{verbatim}
\usepackage{hyperref,color}
\usepackage[capitalize,nameinlink]{cleveref}
\usepackage[dvipsnames]{xcolor}
\hypersetup{
	colorlinks=true,
	pdfpagemode=UseNone,
    citecolor=OliveGreen,
    linkcolor=NavyBlue,
    urlcolor=Magenta,
	pdfstartview=FitW
}
\usepackage{appendix}
\crefname{appsec}{Appendix}{Appendices}
\usepackage{tikz}

\theoremstyle{plain}
\newtheorem{theorem}{Theorem}[section]
\newtheorem{proposition}[theorem]{Proposition}
\newtheorem{lemma}[theorem]{Lemma}

\theoremstyle{definition}
\newtheorem{definition}[theorem]{Definition}

\newtheorem*{assumption*}{Assumption}

\theoremstyle{remark}
\newtheorem{remark}[theorem]{Remark}

\crefname{lemma}{Lemma}{Lemmas}
\crefname{theorem}{Theorem}{Theorems}
\crefname{definition}{Definition}{Definitions}
\crefname{fact}{Fact}{Facts}
\crefname{claim}{Claim}{Claims}
\crefname{proposition}{Proposition}{Propositions}

\newcommand{\E}{\mathbb{E}}
\newcommand{\Var}{\mathrm{Var}}

\newcommand{\Ent}{\mathrm{Ent}}
\newcommand{\one}{\mathbbm{1}}

\newcommand{\ceil}[1]{\left\lceil #1 \right\rceil}

\newcommand{\poly}{\mathrm{poly}}
\newcommand{\dist}{\mathrm{dist}}

\newcommand{\eps}{\varepsilon}

\newcommand{\N}{\mathbb{N}}
\newcommand{\Z}{\mathbb{Z}}

\newcommand{\R}{\mathbb{R}}

\newcommand{\BB}{\mathcal{B}}

\newcommand{\HH}{\mathcal{H}}
\newcommand{\II}{\mathcal{I}}

\newcommand{\QQ}{\mathcal{Q}}

\newcommand{\XX}{\mathcal{X}}
\newcommand{\YY}{\mathcal{Y}}
\newcommand{\ZZ}{\mathcal{Z}}

\newcommand{\dtv}{d_{\mathrm{TV}}}
\newcommand{\DKL}{D_{\mathrm{KL}}}

\newcommand{\diam}{\mathrm{diam}}
\newcommand{\tw}{\mathrm{tw}}

\newcommand{\ball}{\mathsf{B}}

\begin{document}
	
\title{Combinatorial Approach for Factorization of Variance and Entropy in Spin Systems}
\author{Zongchen Chen\thanks{Department of Mathematics, Massachusetts Institute of Technology. Email: \texttt{zongchen@mit.edu}.} }
\date{\today}

\maketitle

\begin{abstract}
We present a simple combinatorial framework for establishing approximate tensorization of variance and entropy in the setting of spin systems (a.k.a.\ undirected graphical models) based on balanced separators of the underlying graph. 
Such approximate tensorization results immediately imply as corollaries many important structural properties of the associated Gibbs distribution, in particular rapid mixing of the Glauber dynamics for sampling.
We prove approximate tensorization by recursively establishing block factorization of variance and entropy with a small balanced separator of the graph.
Our approach goes beyond the classical canonical path method for variance and the recent spectral independence approach, and allows us to obtain new rapid mixing results. 
As applications of our approach, we show that:
\begin{enumerate}
\item On graphs of treewidth $t$, the mixing time of the Glauber dynamics is $n^{O(t)}$, which recovers the recent results of Eppstein and Frishberg \cite{EF22+} with improved exponents and simpler proofs; 

\item On bounded-degree planar graphs, strong spatial mixing implies $\widetilde{O}(n)$ mixing time of the Glauber dynamics, which gives a faster algorithm than the previous deterministic counting algorithm by Yin and Zhang \cite{YZ13}. 

%\item For the hardcore model on the random graph $\GG(n,d/n)$ with fugacity $\lambda < \lambda_c(d)$, the mixing time of the Glauber dynamics is $n^{1+o(1)}$.
%This completes the recent result of Bez{\'a}kov{\'a}, Galanis, Goldberg, {\v{S}}tefankovi{\v{c}} \cite{BGGS22} who proves rapid mixing for certain heat-bath block dynamics.
\end{enumerate}
%In all the applications above, we establish the approximate tensorization of entropy, thus the modified log-Sobolev inequalities and concentration inequalities follow.

%that are widely studied, including the Poincar\'{e} inequality, the modified log-Sobolev inequality, rapid mixing of the Glauber dynamics, and concentration of the measure.

\end{abstract}

\thispagestyle{empty}

%\newpage
%\tableofcontents
%
%\thispagestyle{empty}

\newpage

\setcounter{page}{1}

\section{Introduction}
\label{sec:intro}

Spin systems, also known as undirected graphical models, are important models for describing the joint distribution of interacting random variables. 
Spin systems were first studied in statistical physics but have been widely used in many other areas including computer science, social network, and biology. 
%However, understanding their basic probabilistic properties or performing computational tasks such as sampling, inference, or learning are quit challenging and sometimes computationally intractable. 

Consider a general spin system defined on a graph $G = (V,E)$.
The model describes a distribution, called \emph{Gibbs distribution}, over all spin configurations where each vertex $v$ is assigned a spin $\sigma_v$ from a finite set $\QQ$. 
Given functions $\phi: \QQ \times \QQ \to \R_{\ge 0}$ characterizing pairwise interactions and $\psi: \QQ \to \R_+$ measuring external bias, the Gibbs distribution associated with a spin system is defined as
\begin{equation}\label{eq:gibbs-int}
\mu(\sigma) = \frac{1}{Z} \prod_{e=uv \in E} \phi(\sigma_u,\sigma_v) \prod_{v \in V} \psi(\sigma_v), \quad \forall \sigma: V \to \QQ
\end{equation}
where $Z$ is a normalizing constant known as the \emph{partition function}. 

We mention two classical examples of spin systems which have been extensively studied.
The first is the hardcore model of independent sets. The Gibbs distribution is supported over the collection $\II_G$ of all independent sets of $G$ where each $I \in \II_G$ has density $\mu(I) = \lambda^{|I|}/Z$, where $Z = \sum_{I \in \II_G} \lambda^{|\lambda|}$. Observe that this corresponds to $\QQ = \{0,1\}$, $\phi(\sigma_u,\sigma_v) = 1-\sigma_u \sigma_v$, and $\psi(\sigma_v) = 1+(\lambda-1) \sigma_v$ in \cref{eq:gibbs-int} with $\sigma$ being the indicator vector.
Another example is random vertex colorings where the Gibbs distribution is the uniform distribution over all proper $q$-colorings of $G$. 
This corresponds to $\QQ = [q] = \{1,\dots,q\}$, $\phi(\sigma_u,\sigma_v) = \one\{\sigma_u \neq \sigma_v\}$, and $\psi \equiv 1$ in \cref{eq:gibbs-int}.

We study the problem of sampling from the Gibbs distribution of a spin system. 
In particular, we consider the single-site \emph{Glauber dynamics} (also called Gibbs sampling) which is perhaps the simplest and most popular Markov chain Monte Carlo (MCMC) algorithm for sampling from a high-dimensional distribution. 
The Glauber dynamics is an ergodic Markov chain where in each iteration we choose a vertex $v$ uniformly at random, and update the spin $\sigma_v$ conditional on the spin values of all other vertices. 
The mixing time of Glauber dynamics is the smallest $t$ such that, starting from any initial configuration $\sigma^{(0)}$, the distribution of $\sigma^{(t)}$ after $t$ steps is $1/4$-close to the target distribution $\mu$ in total variation distance.

Bounding the mixing time of Glauber dynamics is challenging even for simplest spin systems like the hardcore model or random colorings. 
%In fact, one of the most important open problem in this area is to show rapid mixing for sampling $q$-colorings on graphs of maximum degree $\Delta$ when $q \ge \Delta+2$ (where Glauber dynamics is ergodic). 
One common method for establishing mixing time bounds of Markov chains is by proving associated functional inequalities such as the Poincar\'{e} inequality or the standard/modified log-Sobolev inequality. 
For a function $f: \QQ^V \to \R$, the expectation of $f$ with respect to a Gibbs distribution $\mu$ is defined as $\E f = \sum_{\sigma: V \to \QQ} \mu(\sigma) f(\sigma)$. The variance and entropy functionals are defined as $\Var f = \E [f^2] - (\E f)^2$, and $\Ent f = \E [f \log f] - (\E f) \log(\E f)$ for non-negative $f$. 
For Glauber dynamics specifically, the Poincar\'{e} inequality can be expressed equivalently in the form of 
\begin{equation}\label{eq:var-AT}
\Var f \le C \sum_{v \in V} \E[\Var_v (f)], \quad \forall f: \QQ^V \to \R
\end{equation}
where on the right-hand side $\Var^\eta_v f$ is the variance of $f$ under the conditional distribution $\mu^\eta_v$ of $\sigma_v$ where all other vertices are fixed to be some $\eta: V \setminus \{v\} \to \QQ$, and $\E[\Var_v f]$ takes expectation over $\eta$.
The inequality \cref{eq:var-AT} is called \emph{Approximate Tensorization (AT) of variance}. 
The Poincar\'{e} inequality, or equivalently AT of variance, implies that Glauber dynamics mixes in $O(Cn^2)$ steps.

One can also consider the entropy analog of \cref{eq:var-AT}, known as \emph{Approximate Tensorization (AT) of entropy}:
\begin{equation}\label{eq:ent-AT}
\Ent f \le C \sum_{v \in V} \E[\Ent_v (f)], \quad \forall f: \QQ^V \to \R.
\end{equation}
AT of entropy is a much stronger property. It implies the modified log-Sobolev inequality for the Glauber dynamics, and gives a sharper mixing time bound $O(Cn \log n)$ (see \cref{lem:AT-Glauber}) which is optimal if $C$ is constant.

Both \cref{eq:var-AT,eq:ent-AT} were explicitly mentioned and carefully studied in \cite{CMT15}, though they were implicitly used in even earlier works, see e.g.\ \cite{Mar99,GZ03,Cesi01}. 
On a high level, both \cref{eq:var-AT,eq:ent-AT} say that the global fluctuation of a function, quantified as variance or entropy, is always controlled by the sum of local fluctuations at each single variable, which is intuitively true when all variables are sufficiently independent from each other. Indeed, one has $C \ge 1$ in \cref{eq:var-AT,eq:ent-AT} with the equality holds iff $\mu$ is a product distribution. See also \cite{BCSV22,KHR22} for recent applications of AT in learning and testing.

Establishing AT of variance and entropy is a challenging task even for simple distributions.
For variance, the canonical path approach is a common way of proving the Poincar\'{e} inequality (i.e., bounding the spectral gap) in the setting of spin systems. 
The high-level idea is to construct a family of canonical paths or more generally a multi-commodity flow between each pair of configurations and then use the congestion of the flow to establish \cref{eq:var-AT}.
Canonical paths have found many successful applications such as matchings \cite{JS89}, ferromagnetic Ising model \cite{JS93}, and bipartite perfect matchings \cite{JSV04}. 
However, constructing canonical paths is not easy at all and usually involves some specific technical complications for each problem. 
%but most applications are limited to Holant-type problems based on line graphs.

Meanwhile, for entropy it is much more difficult to establish AT or other related functional inequalities like the standard/modified log-Sobolev inequality. 
In many cases they are proved analytically, relying on the topology being the lattice \cite{Mar99,GZ03,Cesi01}.
It was also known that for high-temperature models such as under the Dobrushin uniqueness, AT of entropy holds with $C=O(1)$ \cite{CMT15,Marton19}.

Recently, the spectral/entropic independence approach was introduced \cite{ALO20,AJKPV22} and becomes a powerful tool for establishing AT of both variance and entropy. 
For many families of spin systems it achieves $C=O(1)$ in \cref{eq:var-AT,eq:ent-AT} and thus shows optimal $O(n\log n)$ mixing of Glauber dynamics. 
For example, for the hardcore model one obtains AT of variance and entropy with $C=O(1)$ when $\lambda < \lambda_c(\Delta)$ by a sequence of recent works \cite{ALO20,CLV20,CLV21,CFYZ21,AJKPV22,CFYZ22,CE22}, where $\Delta$ denotes the maximum degree and $\lambda_c(\Delta)$ is the tree uniqueness threshold \cite{Wei06}. 
It was known that Glauber dynamics can be exponentially slow when $\lambda > \lambda_c(\Delta)$ \cite{MWW07} and hence $C = e^{\Omega(n)}$ for some graphs. 
The critical value $\lambda_c(\Delta)$ in fact pinpoints a computational phase transition, see \cite{Wei06,Sly10,SS14,GSV16} for more discussions. 
Though the spectral independence approach works well on general graphs for proving optimal mixing time, it does not apply when the mixing time is a larger polynomial instead of nearly linear. 
%The reason is that it does not take into account the structures of the model.
For example, for the hardcore model on trees the Glauber dynamics always mixes in polynomial time for all $\lambda>0$ \cite{JSTV04} but (constant) spectral independence fails for large $\lambda$.

In this paper, we ask if there is a natural and direct way of proving AT beyond canonical paths or spectral independence, especially when we have extra knowledge of the underlying graph structure.
We present a simple combinatorial approach for proving AT of variance and entropy based on the existence of balanced separators of the graph, see \cref{thm:AT-sep-decomp,thm:strong-AT-sep-decomp}. 
Using this approach we are able to obtain new rapid mixing results as immediate corollaries for certain classes of graphs such as bounded-treewidth graphs or planar graphs. 
We are also able to derive many previous results in a simple and straightforward way in contrast to the detailed technical proofs that were known previously. 
For example, we can easily deduce within one page that Glauber dynamics for $q$-colorings on complete $d$-ary trees is rapidly mixing for all $d\ge 2, q\ge 3$, see \cref{prop:coloring-tree}.

Our proof approach is nicely explained for graphs of bounded treewidth.
The treewidth of a graph characterizes the closeness of a graph to a tree, and is an important parameter for getting fixed parameter tractable algorithms for many graph problems. 
For the hardcore model and random colorings on bounded-treewidth graphs, we obtain rapid mixing of Glauber dynamics as immediate consequences of \cref{thm:AT-sep-decomp,thm:strong-AT-sep-decomp}, which improves the results in \cite{EF22+} with better exponents.
We remark that we can also obtain similar results for more general spin systems but for simplicity we only state them for the hardcore model and vertex colorings which are the most commonly studied examples .

\begin{theorem}\label{thm:hardcore-treewidth}
Let $G=(V,E)$ be an $n$-vertex graph of treewidth $t \ge 1$. 
The mixing time of the Glauber dynamics for sampling from the hardcore model on $G$ with fugacity $\lambda > 0$ is $n^{O(1+t\log(1+\lambda))}$.
\end{theorem}

\begin{theorem}\label{thm:main-coloring-treewidth}
Let $G=(V,E)$ be an $n$-vertex graph of maximum degree $\Delta \ge 3$ and treewidth $t \ge 1$.
For any $q \ge \Delta+2$, the mixing time of the Glauber dynamics for sampling uniformly random $q$-colorings of $G$ is $n^{O(t\Delta)}$.
\end{theorem}
Previously, \cite{EF22+} presented mixing time bounds $n^{O(1+t \log\hat{\lambda})}$ for hardcore model where $\hat{\lambda} = \max\{\lambda,1/\lambda\}$ and $n^{O(t\Delta \log q)}$ for random $q$-colorings. 
Our mixing time bounds are better in the exponents and our proof is much simpler avoiding the technical construction of multi-commodity flows as done in \cite{EF22+}.
We note that for bounded-treewidth graphs one can exactly compute the partition function and thus sample in time $e^{O(t)} \cdot \poly(n)$, see e.g.\ \cite{YZ13,WTZL18} and the references therein. 
However, the performance of the Glauber dynamics is still unclear for such graphs which is the focus of this paper.

We establish \cref{thm:hardcore-treewidth,thm:main-coloring-treewidth} by factorizing variance recursively using a balanced separator of the graph; this is inspired by previous works on trees \cite{MSW04} and lattice \cite{Cesi01,CP21}. 
Since $G$ has treewidth $t$, it has a balanced separator of size at most $t$. 
More specifically, there exists a partition $V = A \cup B \cup S$ such that $|S| \le t$, $|A| \le 2n/3$, $|B| \le 2n/3$ and there is no edge between $A$ and $B$. Given such we are able to establish a block factorization of variance
\begin{equation}\label{eq:int-block}
\Var f \le C_0 \left( \E[\Var_A f] + \E[\Var_B f] + \E[\Var_S f] \right),
\end{equation}
where $C_0$ is a constant independent of $n$. Since the size of $S$ is bounded, we can factorize the variance on $S$ into single vertices. Then by recursively applying \cref{eq:int-block} to $A$ and $B$ respectively, we are able to prove AT of variance with constant $C = C_0^{O(\log n)} = n^{O(1)}$. 
We present our universal proof approach for AT in \cref{thm:AT-sep-decomp,thm:strong-AT-sep-decomp}, and summarize basic tools for establishing such block factorization \cref{eq:int-block} in \cref{subsec:tools,subsec:marginal}. 

We note that \cite{EF22+} also used balanced separators to construct multi-commodity flows and obtained similar results. In comparison, our approach establishes AT in a more direct way and is arguably simpler in nature.
We remark that we could also establish AT of entropy, but since we are not aiming to get an optimal exponent in the mixing time we choose to work with variance which is much easier for calculation.

As another application of our approach, we consider planar graphs and show that the \emph{Strong Spatial Mixing (SSM)} property (\cref{def:SSM}) implies nearly optimal mixing time of Glauber dynamics. SSM is an important structural property characterizing the exponential decay of correlations between two subsets of variables as their graph distance grows. There is a rich literature in establishing rapid mixing of Glauber dynamics from SSM but mostly only for special classes of graphs such as lattice, see e.g.\ \cite{Mar99,GZ03,Cesi01,DSVW04,GMP05}. 
Here we consider general bounded-degree planar graphs with no restriction on the underlying topology.
Previously, \cite{YZ13} presented a determinstic counting algorithm under the same assumptions with large polynomial running time. Our result can be viewed as a faster sampling algorithm with nearly linear running time.
We state our result for the hardcore model but it extends easily to general spin systems.

\begin{theorem}\label{thm:SSM-planar}
Let $G=(V,E)$ be an $n$-vertex planar graph of maximum degree $\Delta \ge 3$.
Consider the hardcore model on $G$ with fugacity $\lambda > 0$.
If SSM holds, then the mixing time of the Glauber dynamics is $O(n \log^4n)$. 
\end{theorem}

In fact, we prove \cref{thm:SSM-planar} more generally for all graphs of polynomially bounded local treewidth, see \cref{thm:local-treewidth}. 
This is a class of graphs such that all local balls of radius $r$ have treewidth $\poly(r)$.
This includes for examples bounded-treewidth graphs, planar graphs, or graphs with polynomial neighborhood growth such as regions in $\Z^d$ (see \cref{thm:growth}).
We remark that the result in \cite{YZ13} holds only for graphs of linearly bounded local treewidth and thus our result holds for a larger class of graphs.
We establish AT of entropy by combining SSM and local structure of the underlying graph. 
In particular, we show a new low-diameter decomposition of graphs (see \cref{lem:decomp}) based on a classical result of Linial and Saks \cite{LS93}, 
which allows us to focus on subgraphs of $\poly(\log n)$ diameter assuming SSM.
Note that here we have to work with entropy in order to get an $\widetilde{O}(n)$ mixing time.

\paragraph{Paper organization.} 
After giving preliminaries in \cref{sec:pre}, we present our new framework for establishing approximate tensorization in \cref{sec:approach}. We prove \cref{thm:hardcore-treewidth,thm:main-coloring-treewidth} in \cref{sec:treewidth} for Glauber dynamics on bounded-treewidth graphs. We prove \cref{thm:SSM-planar} more generally for graphs of bounded local treewidth in \cref{sec:SSM}. We present missing proofs of basic tools for approximate tensorization and block factorization in \cref{sec:proofs}.

\paragraph{Acknowledgments.}
The author thanks Kuikui Liu, Eric Vigoda, and Thuy-Duong Vuong for stimulating discussions. The author thanks Eric Vigoda for helpful comments on the manuscript. 

\section{Preliminaries}
\label{sec:pre}

\subsection{Spin systems}

We consider general families of spin systems, also known as undirected graphical models. 
Let $G = (V,E)$ be a graph. 
Every vertex $v \in V$ is associated with a finite set $\QQ_v$ of spins (colors, labels). Let $\XX = \prod_{v\in V} \QQ_v$ denote the product space of all spin assignments called configurations. 
Let $\Phi = \{\phi_e: e \in E\}$ be a collection of edge interactions such that for every edge $e=uv \in E$, $\phi_e: \QQ_u \times \QQ_v \to \R_{\ge 0}$ is a function mapping spins of the two endpoints to a non-negative weight, characterizing the interaction between neighboring vertices.
Let $\Psi = \{\psi_v: v \in V\}$ be a collection of external fields such that for every vertex $v \in V$, $\psi_v: \QQ_v \to \R_{\ge 0}$ assigns a weight to each color representing the bias towards each color. 
The induced Gibbs distribution $\mu = \mu_{G, \Phi, \Psi}$ is given by
\[
\mu(\sigma) = \frac{1}{Z} \prod_{e \in E} \phi_e(\sigma_e) \prod_{v \in V} \psi_v(\sigma_v), \quad \forall \sigma\in \XX
\]
where $\sigma_v$ denotes the color of a vertex $v \in V$ and $\sigma_e$ denotes the (partial) spin assignment of vertices in $e$, and $Z = Z_{G,\Phi,\Psi}$ is the partition function defined as
\[
Z = \sum_{\sigma \in \XX} \prod_{e \in E} \phi_e(\sigma_e) \prod_{v \in V} \psi_v(\sigma_v).
\]
We assume $Z>0$ so that the Gibbs distribution is well-defined. 
%Let $\XX = \{\sigma\in \prod_{v\in V} \QQ_v: \mu(\sigma) > 0\}$ denote the support of $\mu$ which is non-empty.

For a subset $W \subseteq V$ let $\mu_W$ denotes the marginal distribution projected on $W$. 
We say a spin configuration $\eta \in \prod_{v\in\Lambda} \QQ_v$ on some subset $\Lambda \subseteq V$ is a (feasible) pinning if $\mu_\Lambda(\eta) > 0$. 
For any pinning $\eta$ on $\Lambda \subseteq V$, the conditional Gibbs distribution $\mu^\eta$, where the configuration on $\Lambda$ is fixed to be $\eta$, can be viewed as another spin system on the graph $G \setminus \Lambda$. 
For any $W \subseteq V \setminus \Lambda$ we further define $\mu^\eta_W$ to be the marginal on $W$ under the conditional distribution $\mu^\eta$.

The Glauber dynamics is one of the simplest and most popular Markov chain Monte Carlo algorithms for sampling from the Gibbs distribution of a spin system. 
In each step of Glauber dynamics, we pick a vertex $v \in V$ uniformly at random and update the spin value $\sigma_v$ conditional on the configuration $\sigma_{V \setminus \{v\}}$ of all other vertices, i.e., from the conditional marginal $\mu_v^{\sigma_{V\setminus \{v\}}}$.
Two distinct configurations $\sigma,\tau \in \XX$ are said to be adjacent iff they differ in the spin value at a single vertex. 
Hence, for Glauber dynamics configurations only move to adjacent ones. 
We assume that under any pinning $\eta$, the Glauber dynamics for the conditional distribution $\mu^\eta$ is irreducible; namely, one feasible configuration can move to another through a chain of feasible configurations where consecutive pairs are adjacent. 
%We also assume the Gibbs distribution $\mu$ is totally connected.
This is necessary for the Glauber dynamics to be ergodic and is naturally true for many spin systems of interest, including the hardcore model, random $q$-colorings with $q\ge \Delta+2$, etc.

\subsection{Block factorization of variance and entropy}

Consider the Gibbs distribution $\mu$ of a spin system $(G,\Phi,\Psi)$ defined on a graph $G = (V,E)$ with state space $\XX = \prod_{v\in V} \QQ_v$.
%Let $\XX = \{\sigma\in \prod_{v\in V} \QQ_v: \mu(\sigma) > 0\}$ denote the support of $\mu$.

\begin{definition}
%Let $\mu$ be a distribution supported on a finite state space $\XX$ and 
Let $f: \XX \to \R$ be a function.

\begin{itemize}
\item The expectation of $f$ with respect to $\mu$ is defined as 
$\E_\mu f = \sum_{\sigma \in \XX} \mu(\sigma) f(\sigma).$

\item The variance of $f$ with respect to $\mu$ is defined as 
$\Var_\mu f = \E[(f-\E f)^2].$ 

\item For non-negative $f$, the entropy of $f$ with respect to $\mu$ is defined as 
$\Ent_\mu f = \E[f \log f] - (\E f) \log(\E f),$ 
with the convention that $0\log 0 = 0$. 
\end{itemize}
 
\end{definition}

We often omit the subscript and write $\E f,\Var f,\Ent f$ when the underlying distribution $\mu$ is clear from context.

Let $B \subseteq V$ be a subset of vertices. 
For any pinning $\eta$ on $V \setminus B$, the expectation, variance, and entropy of a function $f$ with respect to the conditional Gibbs distribution $\mu^\eta_B$ is denoted as $\E^\eta_B f$, $\Var^\eta_B f$, and $\Ent^\eta_B f$; in particular, we view all of them as functions mapping a full configuration $\sigma \in \XX$ to a real, which depends only on the pinning $\eta = \sigma_{V \setminus B}$.

The following lemma summarizes several basic and important properties of variance and entropy in spin systems. The proof can be found in \cite{MSW04} and the references therein.

\begin{lemma}[{\cite[Eq.\ (3), (4), (5)]{MSW04}}]
\label{lem:MSW04}
Let $f: \XX \to \R_{\ge 0}$ be a function. 
\begin{enumerate}[(1)]
\item For any subsets $B \subseteq A\subseteq V$, it holds that $\E[\Ent_A f] = \E[\Ent_B f] + \E [\Ent_A(\E_B f)]$;
\item For $B = \bigcup_{i=1}^m B_i$ where $B_1,\dots,B_m \subseteq V$ are pairwise disjoint subsets such that every distinct pair of $B_i,B_j$ are disconnected in $G[B]$, it holds that $\E[\Ent_B f] \le \sum_{i=1}^m \E[\Ent_{B_i} f]$;
\item For any subsets $A,B \subseteq V$ such that there are no edges between $A$ and $B \setminus A$, it holds that $\E\left[ \Ent_A (\E_B f) \right] \le \E\left[ \Ent_A (\E_{A \cap B} f) \right]$.
\end{enumerate}
All three properties (1), (2), (3) hold with variance as well.
\end{lemma}

\begin{definition}
Let $\BB$ be a collection (possibly multiset) of subsets of $V$. 
%We say that $\mu$ satisfies \emph{(approximate) $\BB$-factorization of variance} with multiplier $C$ if for every function $f: \XX \to \R$ it holds
%\begin{equation}\label{eq:var-block-factor}
%\Var f \le C \sum_{B \in \BB} \E[\Var_{B} f].
%\end{equation}
%We say that $\mu$ satisfies \emph{(approximate) $\BB$-factorization of entropy} with multiplier $C$ if for every function $f: \XX \to \R_{\ge 0}$ it holds
%\begin{equation}\label{eq:block-factor}
%\Ent f \le C \sum_{B \in \BB} \E[\Ent_{B} f].
%\end{equation}
We say that $\mu$ satisfies \emph{$\BB$-factorization of variance (resp., entropy)} with multiplier $C$ if for every function $f: \XX \to \R$ (resp., $f: \XX \to \R_{\ge 0}$) it holds
\begin{equation}\label{eq:block-factor}
\Var f \le C \sum_{B \in \BB} \E[\Var_{B} f] 
\qquad
\Big(\Big.~ \text{resp.,~~} \Ent f \le C \sum_{B \in \BB} \E[\Ent_{B} f] ~\Big.\Big).
\end{equation}
\end{definition}

\begin{remark}
Following \cite{BGGS22}, we call $C$ the ``multiplier'' instead of ``constant'' since it may depend on $n$. 
%The standard way of writing block factorization of entropy is that
%\[
%\delta(\alpha) \Ent f \le C \sum_{B \in \BB} \alpha_B \E[\Ent_{B} f]
%\]
%where $\delta(\alpha)$ is the minimum covering weight/probability of a vertex (the minimum probability that a randomly chosen block covers a given vertex) and $C$ is an absolute constant independent of $n$.
\end{remark}
\begin{remark}
The block factorization of variance/entropy can be defined more generally for weighted blocks; we refer to \cite{CP21} for more details.
\end{remark}

The $\BB$-factorization corresponds to the heat-bath block dynamics for sampling from $\mu$: In each iteration, we pick $B \in \BB$ uniformly at random and update $\sigma_B$ from the conditional distribution $\mu_B^{\sigma_{V \setminus B}}$.
More specifically, the $\BB$-factorization of variance is equivalent to the Poincar\'{e} inequality of such block dynamics and the $\BB$-factorization of entropy implies the modified log-Sobolev inequality (but the converse is not true). See \cite{CMT15,CP21}.

%%\bigskip
%%\noindent\textbf{Block Dynamics with $\BB$}
%%\begin{enumerate}
%%\item Pick $B \in \BB$ uniformly at random;
%%\item Update $X(B)$ conditioned on $X(V \setminus B)$.
%%\end{enumerate}
%%\bigskip 
%
%\begin{center}
%\noindent\fbox{%
%    \parbox{0.6\textwidth}{%
%        \centering
%        ~\\[0.2em]
%        \textbf{Block Dynamics with $\BB$}
%        \begin{enumerate}
%        \item Pick $B \in \BB$ uniformly at random;
%        \item Update $X(B)$ conditioned on $X(V \setminus B)$.
%        \end{enumerate}
%    }%
%}
%\end{center}

\begin{definition}
%We say that $\mu$ satisfies \emph{approximate tensorization (AT) of variance} with multiplier $C$ if for every function $f: \XX \to \R$ it holds
%\[
%\Var f \le C \sum_{v \in V} \E[\Var_{v} f].
%\]
%We say that $\mu$ satisfies \emph{approximate tensorization (AT) of entropy} with multiplier $C$ if for every function $f: \XX \to \R_{\ge 0}$ it holds
%\[
%\Ent f \le C \sum_{v \in V} \E[\Ent_{v} f].
%\]
We say that $\mu$ satisfies \emph{approximate tensorization (AT) of variance (resp., entropy)} with multiplier $C$ if for every function $f: \XX \to \R$ (resp., $f: \XX \to \R_{\ge 0}$) it holds
\begin{equation}
\Var f \le C \sum_{v \in V} \E[\Var_{v} f] 
\qquad
\Big(\Big.~ \text{resp.,~~} \Ent f \le C \sum_{v \in V} \E[\Ent_{v} f] ~\Big.\Big).
\end{equation}
\end{definition}

Observe that AT is exactly $\BB$-factorization with $\BB = \{\{v\}: v \in V\}$.
Establishing AT with a small $C$ allows us to conclude rapid mixing of Glauber dynamics, see \cite{CLV21} and the references therein.

\begin{lemma}\label{lem:AT-Glauber}
Suppose it holds that $\min_{\sigma \in \XX:\, \mu(\sigma)>0} \mu(\sigma) = e^{-O(n)}$. 
%\begin{itemize}
%\item 
If $\mu$ satisfies AT of variance with multiplier $C$, then the mixing time of Glauber dynamics is $O(Cn^2)$.
%\item 
If $\mu$ satisfies AT of entropy with multiplier $C$, then the mixing time of Glauber dynamics is $O(Cn\log n)$.
%\end{itemize}
%$O(Cn \log(1/\mu_{\min}))$. 
%$O(Cn \log\log(1/\mu_{\min}))$. 
\end{lemma}

%\subsection{Treewidth}
%
%Let $G=(V,E)$ be a graph.
%For any subset $W\subseteq V$ of vertices, the (strong) diameter of $W$, denoted by $\diam(G[W])$, is the largest distance between two vertices from $W$ \emph{in the induced subgraph $G[W]$}. 
%In particular, if $G[W]$ is disconnected then $\diam_G(W) = \infty$. 
%
%
%Recall that for $S \subseteq U \subseteq V$ and $r \in \N$ we define
%\[
%\ball(S,r) = \{v \in V: \dist_G(v,S) \le r\}
%\]
%and
%\[
%\ball_U(S,r) = \{v \in U: \dist_G(v,S) \le r\}
%\]
%to be the portion of the ball around $S$ of radius $r$ inside $U$.

\subsection{Separator decomposition}
%Let $G=(V,E)$ be a graph.
%For a subset $U \subseteq V$, define the induced subgraph $G[U] = (U,E[U])$ to be the graph with vertex set $U$ and edge set
%\[
%G[U] = \{e = \{u,v\} \in E: u,v \in U\}.
%\] 

We prove AT via a divide-and-conquer argument. 
To accomplish this we need the following separator decomposition, slightly modified from \cite{YZ13}. Such ideas also appeared in many previous works to obtain fixed-parameter tractable algorithms in graphs of bounded treewidth. 

\begin{definition}[Separator Decomposition, \cite{YZ13}]
For a graph $G=(V,E)$, a separator decomposition tree $T_{\mathsf{SD}}$ of $G$ is a rooted tree satisfying the following conditions:
\begin{itemize}
\item Every node of $T_{\mathsf{SD}}$ is a pair $(U,S)$ where $U$ is a subset of vertices and $S$ is a separator of $G[U]$;

\item The root node of $T_{\mathsf{SD}}$ is a pair $(V,S_V)$;

\item For every non-leaf node $(U,S)$, the children of $(U,S)$ are connected components of $G[U \setminus S]$ and their separators; 

\item Every leaf of $T_{\mathsf{SD}}$ is a pair $(U,U)$.
\end{itemize}
A separator decomposition tree is said to be balanced if for every internal node $(U,S)$ and every child $(U',S')$ of $(U,S)$, it holds that $|U'| \le 2|U|/3$. 
\end{definition}

\begin{remark}\label{rem:sep-partition}
Observe that all separators $S$ appearing in $T_{\mathsf{SD}}$ form a partition of $V$.
\end{remark}

It is easy to see that the height of a balanced separator decomposition tree is $O(\log n)$.

\begin{lemma}\label{lem:height}
Suppose $G=(V,E)$ is an $n$-vertex graph with $n\ge 2$. 
If $T_{\mathsf{SD}}$ is a balanced separator decomposition tree of $G$, then the height of $T_{\mathsf{SD}}$ is less than $3\log n$.
\end{lemma}

\begin{proof}
Assume $h \ge 3\log n$. 
Let $(U,U)$ be a leaf of $T_{\mathsf{SD}}$ at distance $h$ from the root node $(V,S_V)$. 
Since all separators are balanced, we have $1 \le |U| \le (2/3)^h n \le n^{1-3\log(3/2)} < 1$, contradiction. 
\end{proof}

\section{Combinatorial Approach for Approximate Tensorization}
\label{sec:approach}

\subsection{Approximate tensorization via separator decomposition}

%For $r \in \N$ and $S \subseteq U \subseteq V$, let $N_{U}(S,r)$ be the set of vertices in $U$ that are at distance at most $r$ from $S$. 
Our main step for establishing approximate tensorization of variance and entropy is given by the following proposition. 
Roughly speaking, if we can find a (small) separator $S\subseteq V$ of the underlying graph $G$ whose removal disconnects $G$, then we can factorize variance/entropy into the block $S$ and all connected components of $V \setminus S$, whose sizes are significantly smaller if the separator $S$ is balanced. 
Given a (balanced) separator decomposition tree, we can continue this process for each smaller block, and in the end tensorize into single vertices.

\begin{proposition}
\label{thm:AT-sep-decomp}
Let $(G,\Phi,\Psi)$ be a spin system defined on a graph $G=(V,E)$ with associated Gibbs distribution $\mu$.
Suppose that $T_{\mathsf{SD}}$ is a separator decomposition tree of $G$ satisfying:
\begin{enumerate}
\item (Block Factorization for Decomposition) 
For every node $(U,S)$, there exists $C_{U,S} \ge 1$, such that for any function $f: \XX \to \R$ we have
\begin{equation}\label{eq:decomp-factor}
\E[\Var_U f] \le C_{U,S} \left( \E\left[ \Var_S f \right] + \E\left[ \Var_{U \setminus S} f \right] \right).
\end{equation}
For all leaves $(U,U)$ we take $C_{U,U} = 1$.

\item (Approximate Tensorization for Separators) 
For every node $(U,S)$, there exists $C_{S} \ge 1$, such that for any function $f: \XX \to \R$ we have
\begin{equation}\label{eq:AT-separa}
\E\left[ \Var_S f \right] \le C_{S} \sum_{v \in S} \E[\Var_v f]. 
\end{equation}
\end{enumerate}
Then the Gibbs distribution $\mu$ satisfies approximate tensorization of variance with multiplier $C$ given by
\[
C = \max_{(U,S)} \left\{ C_{S} \prod_{(U',S')} C_{U',S'} \right\},
\]
where the maximum is taken over all nodes of $T_{\mathsf{SD}}$, and the product is over all nodes $(U',S')$ in the unique path from the root $(V,S_V)$ to $(U,S)$. 
Namely, for every function $f: \XX \to \R$ we have
\begin{align*}
\Var f 
%&\le \left( C_1^h C_2 + \sum_{i=1}^h C_1^i C_2 \right) \sum_{v \in V} \E[\Ent_v f] \\
%&\le (h+1) C_1^h C_2 \sum_{v \in V} \E[\Ent_v f],
&\le C \sum_{v \in V} \E[\Var_v f].
\end{align*}
%where $h$ is the height of $T_{\mathsf{SD}}$ (the largest distance from the root to a leaf). 
\end{proposition}

\begin{remark}
The entropy version of \cref{thm:AT-sep-decomp} is also true and the proof is exactly the same with $\Ent(\cdot)$ replacing $\Var(\cdot)$.
\end{remark}

%\begin{remark}
%To apply \cref{thm:AT-sep-decomp} we would want $h=O(\log n)$, which is true for balanced separator decomposition trees.
%%if for each node $(U,S)$ in the separator decomposition, the separator $S$ is a balanced separator of $G[U]$. 
%We would also want the size of the separators to be constant.
%\end{remark}

The proof of \cref{thm:AT-sep-decomp} is straightforwardly applying properties of variance and entropy given in \cref{lem:MSW04}. 
We use it as our basic strategy for obtaining meaningful AT bounds in many applications. 

\begin{proof}[Proof of \cref{thm:AT-sep-decomp}]
The lemma follows by decomposing the variance level by level on the separator decomposition tree $T_{\mathsf{SD}}$ by \cref{lem:MSW04,eq:decomp-factor,eq:AT-separa}. 
More specifically, for the root $(V,S_V)$ we have that
\begin{align*}
\Var f &\le C_{V,S_V} \left( \E\left[ \Var_S f \right] + \E\left[ \Var_{V\setminus S} f \right] \right) \\
&\le C_{V,S_V} C_{S_V} \sum_{v \in S_V} \E\left[ \Var_v f \right] + C_{V,S_V} \sum_{U: \text{ c.c.\ of $G[V\setminus S_V]$}}\E\left[ \Var_U f \right], 
\end{align*}
where every $U$ is a connected component of $G[V \setminus S_V]$ and we can factorize $\E\left[ \Var_{V\setminus S_V} f \right]$ without loss since it is a product distribution (\cref{lem:MSW04}). In particular, $U$ (together with its separator) is a child of $(V,S_V)$ in $T_{\mathsf{SD}}$. 
Continue the process for each child $(U,S_U)$, we obtain
\begin{align*}
\E\left[ \Var_U f \right] 
&\le C_{U,S_U} \left( \E\left[ \Var_{S_U} f \right] + \E\left[ \Var_{U\setminus S_U} f \right] \right) \\
&\le C_{U,S_U} C_{S_U} \sum_{v \in S_U} \E\left[ \Var_v f \right] + C_{U,S_U} \sum_{W: \text{ c.c.\ of $G[U\setminus S_U]$}}\E\left[ \Var_W f \right], 
\end{align*}
where every $W$ is a connected component of $G[U \setminus S_U]$. 
So in the end we obtain
\[
\Var f \le \sum_{v \in V} C_v \E\left[ \Var_v f \right] \le C \sum_{v \in V} \E\left[ \Var_v f \right],
\]
where for each $v$, 
\[
C_v = C_{S} \prod_{(U',S')} C_{U',S'} \le C
\]
where $(U,S)$ is the unique node such that $v \in S$ (see \cref{rem:sep-partition}) and the product runs through all $(U',S')$ on the unique path from $(V,S_V)$ to $(U,S)$.
%if $v \in S_U$ for a non-leaf node $(U,S_U)$ then we have $C_v = C_1^{\ell+1} C_2 \le C_1^h C_2$ where $\ell$ denotes the distance from the root $(V,S_V)$ to $(U,S_U)$, and if if $v \in U$ for a leaf $(U,U)$ then $C_v = C_1^{\ell} C_2 \le C_1^h C_2$.
This shows the lemma.
\end{proof}

\subsection{Tools for factorization of variance and entropy}
\label{subsec:tools}

To apply \cref{thm:AT-sep-decomp}, one needs to establish block factorization of variance/entropy for decomposition \cref{eq:decomp-factor} and approximate tensorization for separators \cref{eq:AT-separa}. 
In this subsection, we summarizes known and gives new results for factorization of variance and entropy in a very general setting, which are useful for showing \cref{eq:decomp-factor,eq:AT-separa}.
The lemmas in this subsection are suitable for establishing AT or block factorization for disjoint blocks; for overlapping blocks we also need \cref{lem:ent_fact_mar} from \cref{subsec:marginal}.

\paragraph{Two-variable factorization with weak correlation}

We first consider AT for two variables.
Let $X$ and $Y$ be two random variables with joint distribution $\pi=\pi_{XY}$, fully supported on finite state spaces $\XX$ and $\YY$ respectively. 
For applications such as proving \cref{eq:decomp-factor}, $X,Y$ would represent a block of vertices, namely $X = \sigma_S$ and $Y = \sigma_{U \setminus S}$, and we consider the joint distribution of $(X,Y) = \sigma_U$ under an arbitrary pinning $\eta$ outside $U$. 

Denote the marginal distribution of $X$ by $\pi_X$, and for $y \in \YY$ let $\pi_X^y$ be the distribution of $X$ conditioned on $Y = y$.
We define the marginal distribution $\pi_Y$ and for $x \in \XX$ the conditional distribution $\pi_Y^x$ in the same way.

%We start with the case that $X,Y$ have sufficiently weak correlations. 
We use $\E f= \E_\pi f$, $\Var f = \Var_\pi f$, and $\Ent f = \Ent_\pi f$ to denote the expectation, variance, and entropy of some function $f: \XX \times \YY \to \R$ or $\R_{\ge 0}$ under the distribution $\pi$, and use $\E_X^y f$, $\Var_X^y f$, and $\Ent_X^y f$ to denote the variance and entropy under the conditional distribution $\pi_X^y$.
As before, $\E[\Var_X(f)]$ and $\E[\Ent_X(f)]$ represent the expectation of $\Var_X^Y f$ and $\Ent_X^Y f$ where $Y$ is chosen from $\pi_Y$.

We first show AT for $\pi$ when the correlation between $X$ and $Y$ is bounded pointwisely. 
The following lemma is implicitly given in \cite{Cesi01,DPP02}. 
Here we give a self-contained, simplified proof with an improved constant.
%avoiding the technically more involved versions as in \cite[Proposition 2.1]{Cesi01} or \cite[Lemma 5.2]{DPP02}.
The lemma below can also be applied to the main results from \cite{Cesi01,DPP02}.

\begin{lemma}[{\cite[Proposition 2.1]{Cesi01} and \cite[Lemma 5.1 \& 5.2]{DPP02}}]
\label{lem:AT-weak-correlation}
Suppose there exists a real $\eps \in[0,1/2)$ such that for all $y \in \YY$, 
\begin{equation}\label{eq:AT-wc-cond}
\left| \frac{\pi_X^y(x)}{\pi_X(x)} - 1 \right| \le \eps, \quad \forall x \in \XX.
\end{equation}
Then we have
\begin{align}
\Var f &\le \left( 1+\frac{\eps}{1-2\eps} \right) \left( \E[\Var_X f] + \E[\Var_Y f] \right), 
\quad \forall f: \XX \times \YY \to \R \label{eq:var-AT-wc-goal}\\
\text{and}\quad \Ent f &\le \left( 1+\frac{\eps}{1-2\eps} \right) \left( \E[\Ent_X f] + \E[\Ent_Y f] \right), 
\quad \forall f: \XX \times \YY \to \R_{\ge 0}. \label{eq:AT-wc-goal}
\end{align}
\end{lemma}

The proof of \cref{lem:AT-weak-correlation} can be found in \cref{subsec:AT-weak-correlation}.

\paragraph{Two-variable factorization with strong correlation}
Our next result allows stronger correlations between $X$ and $Y$, at a cost of larger constants for AT.
%It is a special case of the more general result of Marton \cite{Marton19}. 
%More specifically, Marton shows that if the Dobrushin dependency matrix of $\pi$ has operator norm strictly less than $1$, then AT holds. 
%In our case here we have only two variables so the Dobrushin matrix is a $2\times 2$ matrix with diagonal entries $0$ by definition.
%Thus, it is simple to find the norm of the Dobrushin matrix and to prove AT. 
%In the lemma below, we present a complete proof for convenience of the readers, which produces the same bound if one directly applies Marton's theorem.

\begin{lemma}%[{\cite[Theorem 1.14]{Marton19}}]
\label{lem:AT-strong-correlation}
Suppose there exist reals $\eps_X,\eps_Y \in [0,1]$ with $\eps_X \eps_Y > 0$ such that
\begin{align}
\dtv(\pi_X^y, \pi_X^{y'}) &\le 1-\eps_X, \quad \forall y,y' \in \YY \label{eq:AT-sc-cond-X} \\
\text{and}\quad
\dtv(\pi_Y^x, \pi_Y^{x'}) &\le 1-\eps_Y, \quad \forall x,x' \in \XX. \label{eq:AT-sc-cond-Y}
\end{align}
Then we have
\begin{align}
\Var f &\le \frac{2}{\eps_X + \eps_Y} \left( \E[\Var_X f] + \E[\Var_Y f] \right), 
\quad \forall f: \XX \times \YY \to \R \label{eq:var-AT-sc-goal} \\
\text{and}\quad
\Ent f &\le \frac{4+2\log(1/\pi_{\min})}{\eps_X + \eps_Y} \left( \E[\Ent_X f] + \E[\Ent_Y f] \right), 
\quad \forall f: \XX \times \YY \to \R_{\ge 0} \label{eq:AT-sc-goal}
\end{align}
where $\pi_{\min} = \min_{(x,y) \in \XX\times \YY:\, \pi(x,y)>0} \pi(x,y)$.
\end{lemma}

%The proof of \cref{lem:AT-strong-correlation} can be found in \cref{subsec:crude-bound}.

\begin{remark}
We remark that the constant for entropy factorization \cref{eq:AT-sc-goal} is not optimal since $\log(1/\pi_{\min})$ depends logarithmically on the size of state space, $|\XX\times \YY| = |\XX| \cdot |\YY|$. Getting rid of this is an interesting open question.
\end{remark}

\paragraph{Multi-variable factorization}

Finally, we consider approximate tensorization for multiple variables, which is helpful for establishing \cref{eq:AT-separa} when the size of the separator is bounded.

Let $X_1,\dots,X_n$ be $n$ random variables where $X_i$ is fully supported on finite $\XX_i$ for each $i$.
The joint distribution of $(X_1,\dots,X_n)$, over the product space $\XX = \prod_{i=1}^n \XX_i$ is denoted by $\pi$. 
For two disjoint subsets $A,B \subseteq [n]$ and a partial assignment $x_A \in \prod_{i \in A} \XX_i$ with $\pi(X_A = x_A) > 0$, let $\pi_B^{x_A} = \pi(X_B = \cdot \mid X_A = x_A)$ denote the conditional marginal distribution on $B$ conditioned on that the variables in $A$ are assigned values from $x_A$. 
In particular, $\pi_B$ denotes the marginal on $B$. 
We write $\pi_i$ and $x_i$ when the underlying set is $\{i\}$ for simplicity. 

As before, we write $\Ent_i^{x_{\bar{i}}} f = \Ent_{X_i}^{x_{\bar{i}}} f$ as the conditional entropy of $f$ on $X_i$ given $x_{\bar{i}} = (x_1,\dots,x_{i-1},x_{i+1},\dots,x_n)$ and let $\E[\Ent_i f]$ be its expectation where $x_{\bar{i}}$ is drawn from $\mu_{[n] \setminus i}$; the variance versions are defined in the same way.

\begin{lemma}
\label{lem:AT-crude-bound}
Suppose there exists a real $\eps \in (0,1]$ such that for every $\Lambda \subseteq [n]$ with $|\Lambda| \le n-2$, every $x_\Lambda \in \prod_{i \in \Lambda} \XX_i$ with $\pi_\Lambda(x_\Lambda) > 0$, 
every $i,j \in [n] \setminus \Lambda$ with $i\neq j$, and every $x_i,x'_i \in \XX_i$ with $\pi_i^{x_\Lambda}(x_i) > 0$ and $\pi_i^{x_\Lambda}(x'_i) > 0$, it holds
\[
\dtv( \pi_j^{x_\Lambda,x_i}, \pi_j^{x_\Lambda,x'_i} ) \le 1-\eps. 
\]
Then we have
\begin{align}
\Var f &\le \frac{1}{\eps^{n-1}} \sum_{i=1}^n \E[\Var_i f],
\quad \forall f: \XX \to \R \\
\text{and} \quad 
\Ent f &\le \frac{2+\log(1/\pi_{\min})}{\eps^{n-1}} \sum_{i=1}^n \E[\Ent_i f],
\quad \forall f: \XX \to \R_{\ge 0} 
\end{align}
where $\pi_{\min} = \min_{x \in \XX:\, \pi(x) > 0} \pi(x)$. 
%and $\Ent_i f = \Ent_{X_i} f$ is the conditional entropy of $f$ on $X_i$. 
\end{lemma}

Observe that for $n=2$, \cref{lem:AT-crude-bound} recovers \cref{lem:AT-strong-correlation} for the special case $\eps_X = \eps_Y = \eps$.
The proofs of \cref{lem:AT-strong-correlation,lem:AT-crude-bound} can be found in \cref{subsec:crude-bound}, which are simple applications of the spectral independence approach based on \cite{AL20,ALO20,FGYZ21}.

\subsection{Example}

%
%For graphs of maximum degree $\Delta$, $\mu$ satisfies AT with multiplier $C_1$ given by:
%
%\begin{itemize}
%
%\item complete $d$-ary tree of height $h$ ($\Delta = d + 1$ and $n = 1 + d + \dots + d^h$)  $\Longrightarrow$ $C_1 = e^{O(dh)} = n^{O\left( \frac{d}{\log d} \right)}$
%%\item box of side length $L$ in $\Z^d$ ($\Delta = 2d$ and $n = L^d$) $\Longrightarrow$ $C_1 = \exp\left( O(L^{d-1}) \right)$
%\item planar graph (treewidth $= O(\sqrt{n})$) $\Longrightarrow$ $C_1 = \exp\left( O(\Delta \sqrt{n}) \right)$
%
%\end{itemize}

Here we give a simple example as an application of \cref{thm:AT-sep-decomp} and tools from \cref{subsec:tools}.
We show polynomial mixing time of Glauber dynamics for sampling $q$-colorings on a complete $d$-ary tree for all $d\ge 2$ and $q \ge 3$. 
This was known previously, see \cite{GJK10,LM11,LMP09,TVVY12,SZ17} for even sharper results.
However, \cref{thm:AT-sep-decomp} allows us to establish this fact in a more straightforward manner, avoiding technical complications such as constructing canonical paths or Markov chain decomposition.

\begin{proposition}[\cite{GJK10}, see \cite{LM11,LMP09,TVVY12,SZ17} for sharper results]
\label{prop:coloring-tree}
Let $d\ge 2$ and $q \ge 3$ be integers. 
Suppose $T$ is a complete $d$-ary tree of height $h$, and denote the number of vertices by $n=\sum_{i=0}^h d^i$.
The mixing time of the Glauber dynamics for sampling uniformly random $q$-colorings of $T$ is $n^{O\big( 1+\frac{d}{q\log d} \big)}$.
\end{proposition}

\begin{proof}
We apply \cref{thm:AT-sep-decomp}. 
We may assume that $q \le 3d$ since otherwise rapid mixing follows by standard path coupling arguments \cite{Jbook}. 
The separator decomposition tree $T_{\mathsf{SD}}$ can be obtain from the original tree $T$: every node of $T_{\mathsf{SD}}$ is of the form $(T_v, \{v\})$ where $T_v$ is the subtree rooted at $v$ and the single-vertex set $\{v\}$ is a separator of $T_v$. 
In particular, we have $C_{\{v\}} = 1$ in \cref{eq:AT-separa} for each separator $\{v\}$. 
We claim that $C_{T_v,\{v\}} = e^{O(d/q)}$ in \cref{eq:decomp-factor} for each node $(T_v, \{v\})$. 
To see this, consider two pinnings where $v$ receives colors $c_1$ and $c_2$, and we can couple all children of $v$ with probability $(1-\frac{1}{q-1})^d = e^{-O(d/q)}$ when neither $c_1$ nor $c_2$ is used. Hence, we can couple the whole subtree $T_v \setminus \{v\}$ with the same probability, implying
%the TV distance for the conditional distribution on $T_v \setminus \{v\}$ between $v$ is colored by $c_1$ and $c_2$ is at most $1-e^{-O(d/q)}$, 
\[
\dtv\left( \mu_{T_v \setminus \{v\}}^{v \gets c_1}, \mu_{T_v \setminus \{v\}}^{v \gets c_2} \right) = 1-e^{-O(d/q)}.
\]
Thus, we deduce from \cref{lem:AT-strong-correlation} that $C_{T_v,\{v\}} = e^{O(d/q)}$.
By \cref{thm:AT-sep-decomp} we obtain that AT of variance holds with multiplier
\[
C = \left(e^{O(d/q)}\right)^h = n^{O\left( \frac{d}{q\log d} \right)}.
\]
The mixing time then follows from \cref{lem:AT-Glauber}.
\end{proof}

%\begin{itemize}
%
%%\item For $q$-colorings on a complete $d$-ary tree of height $h$ ($\Delta = d + 1$ and $n = 1 + d + \dots + d^h$)  $\Longrightarrow$ $C = n^{O\left( 1+\frac{d}{q\log d} \right)}$
%%\item box of side length $L$ in $\Z^d$ ($\Delta = 2d$ and $n = L^d$) $\Longrightarrow$ $C_1 = \exp\left( O(L^{d-1}) \right)$
%\item Hardcore model on a bounded-degree planar graph (treewidth $= c^*\sqrt{n}$) $\Longrightarrow$ $C = \exp\left( O_{\Delta,\lambda}(\sqrt{n}) \right)$.
%
%\end{itemize}

\section{Rapid Mixing for Graphs of Bounded Treewidth}
\label{sec:treewidth}

%For graphs of maximum degree $\Delta$, $\mu$ satisfies AT with multiplier $C_1$ given by:
%\begin{itemize}
%\item treewidth $t$ $\Longrightarrow$ $C_1 = n^{O(\Delta t)}$ ~~\href{https://arxiv.org/abs/2111.03898}{Reference}: Definitions 2.20 and 2.21 and Lemma 2.22
%\end{itemize}

In this section we consider graphs of bounded treewidth. 
It is well-known that all bounded-treewidth graphs have a balanced separator decomposition tree with separators of bounded size. 
\begin{lemma}[\cite{RS86,Gruber}]
\label{lem:treewidth}
If $G$ is a graph of treewidth $t$, then there exists a balanced separator decomposition tree $T_{\mathsf{SD}}$ for $G$ such that for every node $(U,S)$ in $T_{\mathsf{SD}}$ it holds $|S| \le t$. 
\end{lemma}

See also \cite{Bod98,Reed03,HW17} for surveys on treewidth.

%\subsection{Factorization into a Single vertex and Remaining Vertices}

\subsection{Hardcore model}

In this subsection we prove \cref{thm:hardcore-treewidth} for the hardcore model via \cref{thm:AT-sep-decomp}.
As a byproduct, we also show $e^{O(\sqrt{n})}$ mixing time of the Glauber dynamics on arbitrary planar graphs with no restriction on the maximum degree, see \cref{prop:hardcore-planar}.

To apply \cref{thm:AT-sep-decomp}, we need to establish block factorization for every decomposition \cref{eq:decomp-factor} and approximate tensorization for all separators \cref{eq:AT-separa}, which are given by the following two lemmas.

\begin{lemma}\label{lem:hardcore-decomp-factor}
Consider the hardcore model on a graph $G=(V,E)$ with fugacity $\lambda >0$.
Let $S\subseteq U \subseteq V$ be subsets with $|S| \le t$. 
For any pinning $\eta$ on $V \setminus U$, the conditional hardcore Gibbs distribution $\mu^\eta_U$ satisfies $\{S,U\setminus S\}$-factorization of variance with constant $C = 2(1+\lambda)^t$.
%\[
%%C = \exp\big( O\big( t \log(1+\lambda)+\log t+\log\Delta \big) \big).
%C = O\left( t\Delta^2 (1+\lambda)^{t+1} \right).
%\] 
In particular, for every function $f: \XX \to \R$ it holds
\begin{equation}\label{eq:hc-dec-fac}
\E[\Var_U f] \le C \left( \E[\Var_S f] + \E[\Var_{U \setminus S} f] \right).
\end{equation}
\end{lemma}

%\begin{proof}
%Observe that the pinning $\eta$ corresponds to the hardcore model on some induced subgraph. Hence we can apply \cref{lem:single-vx-factor} recursively and obtain 
%\begin{align*}
%\E[\Ent]
%\end{align*}
%\end{proof}

%\begin{proof}
%Let $N^+_W(v) = \{v\} \cup N_W(v) = \{v\} \cup \{u \in W: \{u,v\} \in E\}$.
%It suffices to prove $\BB$-factorization with multiplier $C$ for the marginal distribution $\mu^\eta_{N^+_W(v)}$ where $\BB = \{B_1, B_2\}$, $B_1 = \{v\}$, and $B_2 = N_W(v)$;
%then by \cref{thm:block-marginal} we have $\widetilde{\BB}$-factorization with multiplier $C$ for $\mu^\eta_W$ where $\widetilde{\BB} = \{\widetilde{B}_1, \widetilde{B}_2\}$, $\widetilde{B}_1 = \{v\}$, and $\widetilde{B}_2 = W \setminus \{v\}$.
%
%Now, the marginal distribution $\mu^\eta_{N^+_W(v)}$ is on a set of at most $\Delta+1$ vertices. Using the proof of Lemma~4.2 from CLV21 (optimal mixing of Glauber), but replacing the conductance argument with the coupling argument, one can probably show $C = \poly(\Delta)$ for hardcore with $\lambda = \Theta(1)$, and $C = e^{O(\Delta)}$ for Ising with $\beta = \Theta(1)$. 
%
%%\textcolor{red}{To check. Cannot afford $e^{O(\Delta^2)}$ but $\left( \mu^\eta_{N^+_W(v)} \right)_{\min} = e^{-\Theta(\Delta^2)}$; so we need to be careful.}
%%
%%\textcolor{red}{Seems that you don't get this for colorings when $q \le \Delta + 1$ since the state space for $\mu^\eta_{W}$ may not be connected; meanwhile, $q \ge \Delta + 2$ is okay.}
%\end{proof}

\begin{proof}
%We may assume that $\lambda \ge 1/\Delta$ since otherwise we have $O(n\log n)$ mixing by Dobrushin uniqueness.
%Let $N(S) = \{v \in U \setminus S: uv \in E \text{ for some }u \in S\}$ be the neighborhood of $S$ and let $W = S \cup N(S)$. 
%It suffices to prove $\{ S, N(S) \}$-factorization for the marginal distribution $\mu^\eta_W$ on $W$, since then
%by \cref{lem:block-marginal} we have $\{ S, U \setminus S \}$-factorization for $\mu^\eta_U$ as desired.
%To prove entropy factorization for $\mu^\eta_W$ we utilize \cref{lem:AT-strong-correlation}. 
For any pinning $\tau$ on $U \setminus S$, we have
\[
%\min\left\{ \mu^\eta_U(\sigma_S = \vec{0} \mid \tau),\, \mu^\eta_U(\sigma_S = \vec{0} \mid \xi) \right\} 
\mu^{\eta,\tau}_S(\sigma_S = \vec{0})
\ge \frac{1}{(1+\lambda)^{|S|}} \ge \frac{1}{(1+\lambda)^t}.
\]
Hence, for any two pinnings $\tau,\xi$ on $U \setminus S$ we have $\dtv(\mu^{\eta,\tau}_S,\mu^{\eta,\xi}_S) \le 1-(1+\lambda)^{-t}$. 
%%Since the marginal distribution $\mu^\eta_W$ is on a set of at most $\Delta+1$ vertices, 
%Since $|W| \le |S|+|N(S)| \le t(\Delta+1)$, 
%we deduce from standard marginal bounds for hardcore models that for any feasible configuration $\sigma$ on $W$ it holds
%\[
%\mu^\eta_W(\sigma) 
%\ge \left( \min\left\{ \frac{1}{1+\lambda}, \frac{\lambda}{(1+\lambda)^{\Delta+1}} \right\} \right)^{t(\Delta+1)}
%\ge \left( \frac{\lambda}{(1+\lambda)^{\Delta+1}} \right)^{t(\Delta+1)} 
%\ge \frac{1}{e^{t(\Delta+1)^2(1+\lambda)}},
%\]
%where the last inequality follows from
%\[
%\frac{(1+\lambda)^{\Delta+1}}{\lambda} \le e^\Delta \cdot e^{(\Delta+1)\lambda} < e^{(\Delta+1)(1+\lambda)}.
%\]
Therefore, by \cref{lem:AT-strong-correlation} we have that $\mu^\eta_U$ satisfies $\{S,U\setminus S\}$-factorization of variance with constant $C = 2(1+\lambda)^t$.
%\[
%C = 2(1+\lambda)^t.
%\]
%\[
%C = O\left( (1+\lambda)^t \right) \cdot O\big( t(\Delta+1)^2(1+\lambda) \big) 
%%= O\left( (1+\lambda)^t \right) \cdot O\big( t\Delta ( \Delta \lambda + \Delta ) \big) 
%= O\left( t\Delta^2 (1+\lambda)^{t+1} \right).
%\]
%and by \cref{lem:block-marginal} $\mu^\eta_U$ satisfies $\{ S, U \setminus S \}$-factorization with the same constant.
Taking expectation over $\eta$, we also obtain \cref{eq:hc-dec-fac}.
\end{proof}

\begin{lemma}\label{lem:hardcore-AT-sep}
Consider the hardcore model on a graph $G=(V,E)$ with fugacity $\lambda >0$.
Let $S \subseteq V$ be a subset with $|S| \le t$. 
For any pinning $\eta$ on $V \setminus S$, the conditional hardcore Gibbs distribution $\mu^\eta_S$ satisfies approximate tensorization of variance with constant $C=(1+\lambda)^{t-1}$.
In particular, for every function $f: \XX \to \R$ it holds
\begin{equation}\label{eq:hc-AT-sep}
\E[\Var_S f] \le C \sum_{v \in S} \E[\Var_v f].
\end{equation}
\end{lemma}

\begin{proof}
For any subset $\Lambda \subseteq S$ with $|\Lambda| \le |S|-2$ and any pinning $\tau$ on $\Lambda$, let $u,v \in S \setminus \Lambda$ be two distinct vertices and we have
\[
\mu^{\eta,\tau}_{S \setminus \Lambda}(\sigma_v = 0 \mid \sigma_u = 0) \ge \frac{1}{1+\lambda} 
\quad\text{and}\quad
\mu^{\eta,\tau}_{S \setminus \Lambda}(\sigma_v = 0 \mid \sigma_u = 1) \ge \frac{1}{1+\lambda}.
\]
Hence, $\dtv(\mu^{\eta,\tau,u\gets 0}_v, \mu^{\eta,\tau,u\gets 1}_v) \le 1-(1+\lambda)^{-1}$.
It follows from \cref{lem:AT-crude-bound} that $\mu^\eta_S$ satisfies approximate tensorization of variance with constant
\[
C = (1+\lambda)^{|S|-1} \le (1+\lambda)^{t-1}.
\]
Taking expectation over $\eta$, we also obtain \cref{eq:hc-AT-sep}.
\end{proof}

We are now ready to prove \cref{thm:hardcore-treewidth}. 
\begin{proof}[Proof of \cref{thm:hardcore-treewidth}]
%We may assume that $\lambda \ge e^{-\Delta}$ since otherwise approximate tensorization of entropy and rapid mixing of Glauber dynamics follows from \cref{CLV}. 
We deduce the theorem from \cref{thm:AT-sep-decomp}. 
Since the graph has treewidth $t$, there exists a balanced separator decomposition tree $T_{\mathsf{SD}}$ by \cref{lem:treewidth} where all separators have size at most $t$. 
The height of $T_{\mathsf{SD}}$ is $3\log n$ by \cref{lem:height}. 
Block factorization for decomposition \cref{eq:decomp-factor} is shown by \cref{lem:hardcore-decomp-factor} with $C_{U,S} = 2(1+\lambda)^t$ for each node $(U,S)$. 
Approximate tensorization for separators \cref{eq:AT-separa} is shown by \cref{lem:hardcore-AT-sep} with $C_S = (1+\lambda)^{t-1}$ for each separator $S$. 
Thus, we conclude from \cref{thm:AT-sep-decomp} that the hardcore Gibbs distribution $\mu$ satisfies approximate tensorization of variance with multiplier 
\[
C = (1+\lambda)^{t-1} \cdot \left( 2(1+\lambda)^t \right)^{3\log n}
= n^{O(1+t\log(1+\lambda))}.
\]
The mixing time then follows from \cref{lem:AT-Glauber}.
\end{proof}

As a byproduct, we also show $e^{O(\sqrt{n})}$ mixing of Glauber dynamics for the hardcore model on any planar graph. See \cite{BKMP05,Hein20,EF22+} which can obtain similar results.
\begin{theorem}\label{prop:hardcore-planar}
%Let $\Delta \ge 3$ be an integer and $\lambda > 0$ be a real.
Suppose $G$ is an $n$-vertex planar graph. 
The mixing time of the Glauber dynamics for sampling from the hardcore model on $G$ with fugacity $\lambda > 0$ is $(1+\lambda)^{O(\sqrt{n})}$.
\end{theorem}

\begin{proof}
We apply \cref{thm:AT-sep-decomp}. 
It is well-known that every planar graph has a balanced separator $S \subseteq V$ of size $O(\sqrt{|V|})$, such that each connected component of $G[V \setminus S]$ has size at most $2|V|/3$; see \cite{LT79}. 
We can further find balanced separators for each component, and construct a balanced separator decomposition tree $T_{\mathsf{SD}}$ recursively. 
%For every node $(U,S)$, fix a pinning $\eta$ outside $U$, for any pinning $\tau$ on $U \setminus S$ we have
%$\mu^{\eta,\tau}_S(\sigma_S = \vec{0}) \ge (1+\lambda)^{-|S|}$, 
%and hence for any two pinnings $\tau,\xi$ on $U \setminus S$ we have
%\[
%\dtv \left( \mu^{\eta,\tau}_S, \mu^{\eta,\xi}_S \right) \le 1-\frac{1}{(1+\lambda)^{|S|}}.
%\]
%By \cref{lem:AT-strong-correlation} and taking expectation over $\eta$, we obtain that in \cref{eq:decomp-factor}
By \cref{lem:hardcore-decomp-factor}, for every node $(U,S)$ we have block factorization for decomposition \cref{eq:decomp-factor} with constant
\[
C_{U,S} = 2(1+\lambda)^{|S|} = (1+\lambda)^{O\big(\sqrt{|U|}\big)}. 
\]
%Meanwhile, for each separator $S$, \cref{lem:AT-crude-bound} (taking $\eps = 1/(1+\lambda)$) implies that in \cref{eq:AT-separa}
By \cref{lem:hardcore-AT-sep}, for every separator $S$ we have approximate tensorization for separators \cref{eq:AT-separa} with constant
\begin{equation*}%\label{eq:CS-hardcore}
C_{S} = (1+\lambda)^{|S|-1} = (1+\lambda)^{O(\sqrt{n})}.
\end{equation*}
Therefore, we obtain from \cref{thm:AT-sep-decomp} that AT of variance holds with multiplier
\begin{align*}
C %&= \left( \max_{\text{path $\{(U_i,S_i)\}_{i=1}^\ell$}} \prod_{i=1}^\ell C_1(U_i,S_i) \right) \cdot C_2 \\
%&= \exp\left( \sum_{i=0}^\infty O_{\Delta,\lambda}\left( \sqrt{\left(\tfrac{2}{3}\right)^i n} \right) \right) \cdot \exp( O_{\Delta,\lambda}(\sqrt{n}) ) \\
\le (1+\lambda)^{O(\sqrt{n})}
\cdot
\prod_{i=0}^{\infty} (1+\lambda)^{O\left(\sqrt{\left( \frac{2}{3} \right)^i n}\right)} 
= (1+\lambda)^{O(\sqrt{n})}.
\end{align*}
%This shows that the mixing time of the Glauber dynamics is $\exp( O_{\Delta,\lambda}(\sqrt{n}) )$ for all $\Delta \ge 3$ and $\lambda > 0$.
The mixing time then follows from \cref{lem:AT-Glauber}.
\end{proof}

\subsection{List colorings}

In this subsection we prove \cref{thm:main-coloring-treewidth} from the introduction for colorings.
We consider the more general setting of list colorings, where each vertex $v$ is associated with a list $L_v$ of available colors, and every list coloring assigns to each vertex a color from its list such that adjacent vertices receive distinct colors. 
The Glauber dynamics is ergodic for list colorings if $|L_v| \ge \deg_G(v)+2$ for all $v$. 
One handy feature of list colorings is that any pinning $\eta$ on a subset $\Lambda \subseteq V$ of vertices induces a new list coloring instance on the induced subgraph $G[V \setminus \Lambda]$, for which $|L^\eta_v| \ge \deg_{G[V \setminus \Lambda]}(v)+2$ still holds for all unpinned $v$ where $L^\eta_v$ is the new list of available colors conditioned on $\eta$; see, e.g., \cite{GKM15}.

\begin{theorem}\label{thm:coloring-treewidth}
%Consider list colorings on a graph $G=(V,E)$ of maximum degree $\Delta \ge 3$ where each vertex $v \in V$ is associated with a color list $L_v$ such that $\deg_G(v) + 2 \le |L_v| \le q$. 
%Let $t \ge 1$ denote the treewidth of $G$. 
Let $G=(V,E)$ be a graph of maximum degree $\Delta \ge 3$ and treewidth $t \ge 1$. 
Suppose each vertex $v \in V$ is associated with a color list $L_v$ such that $\deg_G(v) + 2 \le |L_v| \le q$. 
%Then the uniform distribution $\mu$ over all list colorings of $G$ satisfies approximate tensorization of entropy with multiplier $C = n^{O(t\Delta+t\log q)}$. 
%Further, the mixing time of the Glauber dynamics is $n^{O(t\Delta+t\log q)}$.
The mixing time of the Glauber dynamics for sampling uniformly random list colorings of $G$ is $n^{O(t(\Delta+\log q))}$. 
\end{theorem}

\subsubsection{A generalized version of \texorpdfstring{\cref{thm:AT-sep-decomp}}{Proposition 3.1}}

For list colorings we are not able to directly apply \cref{thm:AT-sep-decomp} to prove \cref{thm:coloring-treewidth}. 
Unlike the hardcore model where a vertex can always be unoccupied under any pinning of its neighbors, the lack of such ``universal'' color makes it hard to establish block factorization for decomposition \cref{eq:decomp-factor} using tools from \cref{subsec:tools}.
In this subsection we present a stronger version of \cref{thm:AT-sep-decomp} which allows us to establish block factorization more easily for list colorings. 
Our applications in \cref{sec:SSM} also requires this more general version.

For subsets $S \subseteq U \subseteq V$ and an integer $r \ge 0$, define
\[
\ball(S,r) = \{v \in V: \dist_G(v,S) \le r\}
\]
to be the ball around $S$ of radius $r$ in $G$, 
and define
\[
\ball_U(S,r) = U \cap \ball(S,r) = \{v \in U: \dist_G(v,S) \le r\}
\]
to be the portion of the ball contained in $U$. 
%$\ball_U(S,r) = \{v \in U: \dist_{G[U]}(S,v) \le r\}$ be the ball around $S$ of radius $r$ in the induced subgraph $G[U]$. 
In \cref{thm:strong-AT-sep-decomp} below, we replace $S$ by $\ball_U(S,r)$ in all suitable places in \cref{thm:AT-sep-decomp}, which makes it easier for us to establish block factorization for decomposition \cref{eq:st-decomp-factor} for hard-constraint models like list colorings.

\begin{proposition}\label{thm:strong-AT-sep-decomp}
Let $(G,\Phi,\Psi)$ be a spin system defined on a graph $G=(V,E)$ with associated Gibbs distribution $\mu$.
Let $r \ge 0$ be an integer.
Suppose that $T_{\mathsf{SD}}$ is a separator decomposition tree of $G$ satisfying
\begin{enumerate}
\item (Block Factorization for Decomposition) 
For every node $(U,S)$, there exists $C_{U,S} \ge 1$, such that for any function $f: \XX \to \R_{\ge 0}$ we have
\begin{equation}\label{eq:st-decomp-factor}
\E[\Ent_U f] \le C_{U,S} \left( \E\left[ \Ent_{\ball_U(S,r)} f \right] + \E\left[ \Ent_{U \setminus S} f \right] \right).
\end{equation}
For all leaves $(U,U)$ we take $C_{U,U} = 1$.

\item (Approximate Tensorization for Separators) 
For every node $(U,S)$, there exists $C_{S} \ge 1$, such that for any function $f: \XX \to \R_{\ge 0}$ we have
\begin{equation}\label{eq:st-AT-separa}
\E\left[ \Ent_{\ball_U(S,r)} f \right] \le C_{S} \sum_{v \in \ball_U(S,r)} \E[\Ent_v f]. 
\end{equation}

\item (Bounded Coverage)
There exists $A \ge 1$ such that for any vertex $v \in V$,
\begin{equation}\label{eq:st-a}
|\{(U,S): v \in \ball_U(S,r)\}| \le A. 
\end{equation}
\end{enumerate}
Then the Gibbs distribution $\mu$ satisfies approximate tensorization of entropy with multiplier $C$ given by
\[
C = A \cdot \max_{(U,S)} \left\{ C_{S} \prod_{(U',S')} C_{U',S'} \right\},
\]
where the maximum is taken over all nodes of $T_{\mathsf{SD}}$, and the product is over all nodes $(U',S')$ in the unique path from the root $(V,S_V)$ to $(U,S)$. 
Namely, for every function $f: \XX \to \R_{\ge 0}$ we have
\begin{align*}
\Ent f 
%&\le \left( C_1^h C_2 + \sum_{i=1}^h C_1^i C_2 \right) \sum_{v \in V} \E[\Ent_v f] \\
%&\le (h+1) C_1^h C_2 \sum_{v \in V} \E[\Ent_v f],
&\le C \sum_{v \in V} \E[\Ent_v f].
\end{align*}
%where $h$ is the height of $T_{\mathsf{SD}}$ (the largest distance from the root to a leaf). 
\end{proposition}

\begin{remark}
The variance version of \cref{thm:strong-AT-sep-decomp} is also true and the proof is exactly the same.
\end{remark}

\begin{remark}
Observe that if $r=0$ then $\ball_U(S,r) = S$ and we have $A=1$ in \cref{eq:st-a}, so we recover exactly \cref{thm:AT-sep-decomp}.
\end{remark}

\begin{proof}[Proof of \cref{thm:strong-AT-sep-decomp}]
The proof is the same as for \cref{thm:AT-sep-decomp} by decomposing the entropy level by level on the separator decomposition tree $T_{\mathsf{SD}}$. 
We similarly obtain 
\[
\Ent f \le \sum_{v \in V} C_v \E\left[ \Ent_v f \right], 
%\le C \sum_{v \in V} \E\left[ \Ent_v f \right],
\]
where for each $v$, 
\[
C_v = \sum_{(U,S):\, v \in \ball_U(S,r)} C_{S} \prod_{(U',S')} C_{U',S'} \le C
\]
where the product runs through all $(U',S')$ on the unique path from $(V,S_V)$ to $(U,S)$.
This establishes the proposition.
\end{proof}

The following lemma is helpful for establishing bounded coverage \cref{eq:st-a} when applying \cref{thm:strong-AT-sep-decomp}.

\begin{lemma}\label{lem:A}
Under the assumptions of \cref{thm:strong-AT-sep-decomp}:
\begin{itemize}
\item If $G$ has maximum degree $\Delta \ge 3$, then $A \le \Delta\sum_{i=0}^{r-1} (\Delta-1)^i = O(\Delta^r)$;
\item If $n\ge 3$ and $T_{\mathsf{SD}}$ is a balanced separator decomposition tree, then $A \le 4\log n$. 
\end{itemize}
\end{lemma}

\begin{proof}
If $G$ has maximum degree $\Delta \ge 3$, we have for any $v \in V$ that
\begin{align*}
|\{(U,S): v \in \ball_U(S,r)\}|
&\le |\{(U,S): \dist_{G}(v,S) \le r\}| \\
&\le |\{u \in V: \dist_{G}(v,u) \le r\}| \\
&\le \Delta\sum_{i=0}^{r-1} (\Delta-1)^i = O(\Delta^r),
\end{align*}
where the second inequality is because all separators form a partition of $V$ (see \cref{rem:sep-partition}).

If $T_{\mathsf{SD}}$ is balanced, then the height $h$ of $T_{\mathsf{SD}}$ is at most $3\log n$ by \cref{lem:height}. 
Hence, for any $v \in V$ we have
\[
|\{(U,S): v \in \ball_U(S,r)\}|
\le |\{(U,S): v \in U\}|
\le h+1 \le 4\log n,
\]
as claimed.
\end{proof}

\subsubsection{Block factorization via marginal distributions}
\label{subsec:marginal}

Note that to apply \cref{thm:strong-AT-sep-decomp}, we need to establish the block factorization for decomposition \cref{eq:st-decomp-factor} where the two blocks $\ball_U(S,r)$ and $U \setminus S$ overlap with each other. 
In particular, we can no longer apply \cref{lem:AT-strong-correlation,lem:AT-weak-correlation} directly since fixing the spin assignment in one block greatly changes the conditional distribution on the other block because of the overlap. 
Instead, we show block factorization for the marginal distribution $\mu_{S \cup (U \setminus \ball_U(S,r))}$ for the two blocks $S$ and $U \setminus \ball_U(S,r)$, where we essentially exclude the overlap part. 
It turns out these two notions of block factorization are equivalent to each other, and for the latter we are able to apply tools from \cref{subsec:tools} since the blocks are now disjoint. 
 
We show this equivalence in a general setting. 
Let $X,Y,Z$ be three random variables taking values from finite state spaces $\XX,\YY,\ZZ$ respectively. Their joint distribution is denoted by $\pi = \pi_{XYZ}$.
Denote the marginal distribution for $(X,Y)$ by $\pi_{XY}$ and similarly for other choices of subsets of variables.
We establish block factorization for $\pi$ into two blocks $\{X,Z\}$ and $\{Y,Z\}$ from approximate tensorization for the marginal distribution $\pi_{XY}$. 
%To be more precise, we consider two approximate factorizations of entropy, for $\pi$ and $\pi_{XY}$ respectively. 
%The first one is that, for every function $f: \XX \times \YY \times \ZZ \to \R_{\ge 0}$, it holds
More precisely, we say $\pi$ satisfies $\{\{X,Z\}, \{Y,Z\}\}$-factorization of entropy with constant $C$ if for every function $f: \XX \times \YY \times \ZZ \to \R_{\ge 0}$, it holds
\begin{equation}
\label{eq:XYZ-block}
\Ent f \le C \left( \E\left[ \Ent_{XZ} f \right] + \E\left[ \Ent_{YZ} f \right] \right).
\end{equation}
%The second one is that, for every function $g: \XX \times \YY \to \R_{\ge 0}$, it holds
We say $\pi_{XY}$ satisfies $\{\{X\}, \{Y\}\}$-factorization (i.e., approximate tensorization) of entropy with constant $C$ if for every function $g: \XX \times \YY \to \R_{\ge 0}$, it holds
\begin{equation}
\label{eq:XYZ-marginal-AT}
\Ent_{XY} g \le C \left( \E_{XY}\left[ \Ent_X g \right] + \E_{XY}\left[ \Ent_Y g \right] \right),
\end{equation}
where the reference distribution is the marginal distribution $\pi_{XY}$ over the state space $\XX \times \YY$ and in particular $\Ent_X g$ will be viewed as a function of $Y$ (instead of $Y,Z$). 
The variance versions of \cref{eq:XYZ-block,eq:XYZ-marginal-AT} are defined in the same way with $\Var(\cdot)$ replacing $\Ent(\cdot)$.

We show that these two notions of block factorization of entropy are equivalent to each other.
\begin{lemma}
\label{lem:ent_fact_mar}
The joint distribution $\pi$ satisfies $\{\{X,Z\}, \{Y,Z\}\}$-factorization of entropy (resp., variance) with constant $C$ if and only if the marginal distribution $\pi_{XY}$ satisfies $\{\{X\}, \{Y\}\}$-factorization (i.e., approximate tensorization) of entropy (resp., variance) with constant $C$. 
\end{lemma}
%This is almost obvious if one understands this fact through the corresponding heat-bath block dynamics.
%Recall that, AT for $\pi_{XY}$ corresponds to the following single-site Glauber dynamics: in each step with probability $1/2$ update $X$ according to the conditional marginal distribution $\pi_{XY}(\cdot | Y)$ given the current value of $Y$, and with probability $1/2$ update $Y$ from $\pi_{XY}(\cdot | X)$.
%Meanwhile, the block factorization of entropy for $\{\{X,Z\}, \{Y,Z\}\}$ corresponds to the following heat-bath block dynamics: 
%in each step with probability $1/2$ update $\{X,Z\}$ conditioned on the value of $Y$, i.e., from the distribution $\pi(\cdot | Y)$, and with probability $1/2$ update ${Y,Z}$ from $\pi(\cdot |X)$. 
%It is almost immediate that the two Markov chains are equivalent to each other under a trivial coupling, though the underlying state spaces are not the same.
%In this subsection we show this fact by proving the equivalence of two entropy factorization \cref{eq:XYZ-block,eq:XYZ-marginal-AT}. 
%The benefits of doing so is to allow us to establish \cref{eq:XYZ-block} by proving \cref{eq:XYZ-marginal-AT} instead which is suitable for applying \cref{lem:AT-strong-correlation,lem:AT-weak-correlation}. 
%Note that one cannot directly apply \cref{lem:AT-strong-correlation,lem:AT-weak-correlation} to show \cref{eq:XYZ-block} since the two variables $(X,Z)$ and $(Y,Z)$ are strongly correlated because of $Z$; for example, it is possible that $\dtv(\pi_{XZ}^{yz}, \pi_{XZ}^{y'z'}) =1$ for $(y,z), (y',z') \in \YY \times \ZZ$ when $z\neq z'$. 

The proof of \cref{lem:ent_fact_mar} can be found in \cref{subsec:ent_fact_mar}.

\subsubsection{Proof of \texorpdfstring{\cref{thm:coloring-treewidth}}{Theorem 4.5}}
%To prove \cref{thm:coloring-treewidth}, 
We apply \cref{thm:strong-AT-sep-decomp} with $r=1$, 
where we take $A=\Delta+1$ in \cref{eq:st-a} by \cref{lem:A}. 
We establish block factorization for decomposition \cref{eq:st-decomp-factor} and approximate tensorization for separators \cref{eq:st-AT-separa} in the following two lemmas.
%Thus, $A = O(\Delta)$ by \cref{lem:A}. 
%Also by \cref{lem:AT-crude-bound} we have $C_{S} = O_{\Delta,q,t}(1)$ in \cref{eq:st-AT-separa} for each separator $S$ since $|\ball_U(S,1)| = O(t\Delta)$ by \cref{lem:treewidth}. 
%We focus on \cref{eq:st-decomp-factor}. 

\begin{lemma}\label{lem:coloring-decomp-factor}
Consider list colorings on a graph $G=(V,E)$ of maximum degree $\Delta \ge 3$ where each vertex $v \in V$ is associated with a color list $L_v$ such that $\deg_G(v) + 2 \le |L_v| \le q$. 
Let $S\subseteq U \subseteq V$ be subsets with $|S| \le t$. 
For any pinning $\eta$ on $V \setminus U$, the uniform distribution $\mu^\eta_U$ over list colorings conditioned on $\eta$ satisfies $\{\ball_U(S,1), U\setminus S\}$-factorization of variance with constant $C = 2(2^\Delta q)^t$.
%\[
%C = \exp\big( O( t\Delta+t\log q ) \big).
%\] 
In particular, for every function $f: \XX \to \R$ it holds
\begin{equation}\label{eq:coloring-dec-fac}
\E[\Var_U f] \le C \left( \E[\Var_{\ball_U(S,1)} f] + \E[\Var_{U \setminus S} f] \right).
\end{equation}
%Let $G=(V,E)$ be a graph of maximum degree $\Delta \ge 3$ and treewidth $t \ge 1$. 
%Suppose each vertex $v \in V$ is associated with a color list $L_v$ such that $\deg_G(v) + 2 \le |L_v| \le q$. 
\end{lemma}

\begin{proof}
Let $T = U \setminus \ball_U(S,1)$. 
We first prove $\{ S, T \}$-factorization for the marginal distribution $\mu^\eta_{S \cup T}$ on $S \cup T$ via \cref{lem:AT-strong-correlation}.
%, since then by \cref{lem:block-marginal} we have $\{ S, U \setminus S \}$-factorization for $\mu^\eta_U$ as desired.
%To prove entropy factorization for $\mu^\eta_W$ we utilize \cref{lem:AT-strong-correlation}. 
Let $\tau$ is an arbitrary pinning on $T$ which is feasible under $\eta$.
For any partial list coloring $\sigma$ on $S$ that is feasible under $\eta$, we claim that
\begin{equation}\label{eq:coloring-claim}
%\min\left\{ \mu^{\eta,\tau}_S(\sigma),\, \mu^{\eta,\xi}_S(\sigma) \right\} 
\mu^{\eta,\tau}_S(\sigma)
\ge \frac{1}{(2^\Delta q)^{|S|}} \ge \frac{1}{(2^\Delta q)^t}.
\end{equation}
Note that $\sigma$ must also be feasible under $\tau$ since $S$ and $T$ are not adjacent, and we can extend $\sigma \cup \tau \cup \eta$ to a full list coloring by the assumption $|L_v| \ge \deg_G(v)+2$ for all $v$.
By the chain rule, it suffices to show that for any $\Lambda \subseteq S$ and $v \in S \setminus \Lambda$ we have
\begin{equation}\label{eq:mag-coloring}
\mu^{\eta,\tau,\sigma_\Lambda}_v(\sigma_v) \ge \frac{1}{2^\Delta q}.
\end{equation}
Consider a resampling procedure: starting from a random full list coloring generated from $\mu^{\eta,\tau,\sigma_\Lambda}$, we first resample all (free) neighbors of $v$ and then $v$. 
When resampling neighbors, the probability of avoiding color $\sigma_v$ is at least $1/2^\Delta$ since there are at least two available colors for each neighbor by assumption, and when resampling $v$ the probability of getting $\sigma_v$ is at least $1/|L_v| \ge 1/q$, thus establishing \cref{eq:mag-coloring} and consequently \cref{eq:coloring-claim}. 

We deduce from \cref{eq:coloring-claim} that for any two pinnings $\tau,\xi$ on $T$, it holds 
\[
\dtv(\mu^{\eta,\tau}_S,\mu^{\eta,\xi}_S) \le 1-\frac{1}{(2^\Delta q)^t}.
\]
%Meanwhile, since $|W| = |S|+|N_2(S)| \le t + t\Delta(\Delta-1) \le t\Delta^2$, 
%we deduce from the marginal bound above that for any list coloring $\sigma$ on $W$ it holds
%\[
%\mu^\eta_W(\sigma) 
%\ge \frac{1}{(2^\Delta q)^{t\Delta^2}}
%\ge \frac{1}{2^{t\Delta^2(\Delta+\log_2 q)}}.
%\]
Therefore, by \cref{lem:AT-strong-correlation} we have that $\mu^\eta_{S\cup T}$ satisfies $\{S,T\}$-factorization of variance with constant $C = 2(2^\Delta q)^t$. 
%\[
%C = 2(2^\Delta q)^t.
%%\cdot O\big( t\Delta^2(\Delta+\log_2 q) \big) 
%%= O\left( (1+\lambda)^t \right) \cdot O\big( t\Delta ( \Delta \lambda + \Delta ) \big) 
%%= \exp\big( O( t(\Delta+\log q) ) \big).
%\]
Finally, by \cref{lem:ent_fact_mar} we conclude that $\mu^\eta_U$ satisfies $\{ \ball_U(S,1), U \setminus S \}$-factorization of variance with the same constant, and \cref{eq:coloring-dec-fac} follows by taking expectation over $\eta$.
\end{proof}

\begin{lemma}\label{lem:coloring-at-sep}
Consider list colorings on a graph $G=(V,E)$ of maximum degree $\Delta \ge 3$ where each vertex $v \in V$ is associated with a color list $L_v$ such that $\deg_G(v) + 2 \le |L_v| \le q$. 
Let $B\subseteq V$ be a subset with $|B| \le k$. 
For any pinning $\eta$ on $V \setminus B$, the uniform distribution $\mu^\eta_B$ over list colorings conditioned on $\eta$ satisfies approximate tensorization of variance with constant $C = (2^\Delta q)^{k-1}$.
%\[
%C = \exp\big( O( t\Delta+t\log q ) \big).
%\] 
In particular, for every function $f: \XX \to \R$ it holds
\begin{equation}\label{eq:coloring-AT-sep}
\E[\Var_B f] \le C \sum_{v \in B} \E[\Var_v f].
\end{equation}
%Let $G=(V,E)$ be a graph of maximum degree $\Delta \ge 3$ and treewidth $t \ge 1$. 
%Suppose each vertex $v \in V$ is associated with a color list $L_v$ such that $\deg_G(v) + 2 \le |L_v| \le q$. 
\end{lemma}

\begin{proof}
Fix a pinning $\tau$ on some subset $\Lambda \subseteq B$ and consider two distinct vertices $u,v \in B \setminus \Lambda$. 
For any two colors $c_1,c_2$ that are available to $u$, there exists a color $c$ that is always available to $v$ no matter $\sigma_u = c_1$ or $c_2$. 
This is because: either all neighbors of $v$ are pinned and we can choose any color $c$ available to $v$; or $v$ has at least one free neighbor and thus at least three available colors, and we choose one of them which is not $c_1$ or $c_2$. 
For this color $c$, we have $\mu^{\eta,\tau,u\gets c_j}_v(c) \ge (2^\Delta q)^{-1}$ for $j=1,2$, which was already shown in the proof of \cref{lem:coloring-decomp-factor} (see \cref{eq:mag-coloring}). 
This implies that 
\[
\dtv\left( \mu^{\eta,\tau,u\gets c_1}_v, \mu^{\eta,\tau,u\gets c_2}_v \right) \le 1-(2^\Delta q)^{-1}.
\]
The lemma then follows from an application of \cref{lem:AT-crude-bound}, and taking expectation over $\eta$ we obtain \cref{eq:coloring-AT-sep} as claimed.
\end{proof}

We are now ready to prove \cref{thm:coloring-treewidth} and \cref{thm:main-coloring-treewidth} from the introduction.

\begin{proof}[Proof of \cref{thm:coloring-treewidth}]
%We may assume that $\lambda \ge e^{-\Delta}$ since otherwise approximate tensorization of entropy and rapid mixing of Glauber dynamics follows from \cref{CLV}. 
We deduce the theorem from \cref{thm:strong-AT-sep-decomp}. 
Since the graph has bounded treewidth, there exists a balanced separator decomposition tree $T_{\mathsf{SD}}$ by \cref{lem:treewidth} such that each separator has size at most $t$. 
The height of $T_{\mathsf{SD}}$ is at most $3\log n$ by \cref{lem:height}. 
We pick $r=1$ and hence we can take $A = \Delta+1$ in \cref{eq:st-a} by \cref{lem:A}. 
Block factorization for decomposition \cref{eq:st-decomp-factor} is shown by \cref{lem:coloring-decomp-factor} with $C_{U,S} = 2(2^\Delta q)^t$ for each node $(U,S)$. 
Approximate tensorization for separators \cref{eq:st-AT-separa} follows from \cref{lem:coloring-at-sep} with $C_S = (2^\Delta q)^{(\Delta+1)t-1}$ for each separator $S$, noting that $|\ball_U(S,1)| \le (\Delta+1)t$. 
Thus, we conclude from \cref{thm:strong-AT-sep-decomp} that the uniform distribution $\mu$ of list colorings satisfies approximate tensorization of variance with multiplier 
\[
C = (\Delta+1) \cdot (2^\Delta q)^{(\Delta+1)t-1} \cdot \left( 2(2^\Delta q)^t \right) ^{3\log n}
= n^{O( t(\Delta+\log q) )}.
\]
The mixing time then follows from \cref{lem:AT-Glauber}.
\end{proof}

\begin{proof}[Proof of \cref{thm:main-coloring-treewidth}]
If $q > 2\Delta$ then optimal mixing of Glauber dynamics follows from standard path coupling arguments \cite{Jbook}.
Otherwise, the theorem follows from \cref{thm:coloring-treewidth}.
\end{proof}

\section{Rapid Mixing via SSM for Graphs of Bounded Local Treewidth}
\label{sec:SSM}

%For graphs of maximum degree $\Delta$, $\mu$ satisfies AT with multiplier $C_1$ given by:
%\begin{itemize}
%\item SSM + bounded local treewidth $\Longrightarrow$ $C_1 = \poly(n)$ (or $\widetilde{O}(1)$?)
%%\item SSM + planar graph $\Longrightarrow$ $C_1 = \poly(n)$ (or $\widetilde{O}(1)$?) ~~\href{https://arxiv.org/abs/1207.3564}{Reference}
%%\item SSM + amenable graph $\Longrightarrow$ $C_1 = \poly(n)$ (or $\widetilde{O}(1)$?) (Remove?)
%\end{itemize}

In this section we prove \cref{thm:SSM-planar} from the introduction. 
We consider families of graphs of bounded local treewidth, which include planar graphs as a special case. 
We establish rapid mixing of Glauber dynamics under SSM for all such families. 

Graphs of bounded local treewidth are defined as follows. 

\begin{definition}[Bounded Local Treewidth]
Let $G=(V,E)$ be a graph and $a,d > 0$ be reals. 
We say $G$ has polynomially bounded local treewidth if it satisfies the following diameter-treewidth property: 
for any subgraph $H$ of $G$, it holds
\[
\tw(H) \le a \cdot (\diam(H))^d. 
\]
\end{definition}

Examples of families of graphs that have bounded local treewidth include:
\begin{itemize}
\item Graphs of bounded treewidth;
\item Planar graphs, see \cite{Epp99,DH04,YZ13};
\item Graphs of bounded growth, see \cref{subsec:growth}.
\end{itemize}

The strong spatial mixing property characterizes the decay of correlations in a quantitative way.  
Roughly speaking, it says that in a spin system the correlation between the spin on a vertex $v$ and the configuration on a subset $W$ decays exponentially fast with their graph distance, and such decay holds uniformly under any pinning on any subset of vertices. 

\begin{definition}[Strong Spatial Mixing (SSM)]
\label{def:SSM}
Consider a spin system on a graph $G=(V,E)$ with Gibbs distribution $\mu$. Let $C>0$ and $\delta \in (0,1)$ be reals. 
We say the spin system satisfies the strong spatial mixing property with exponential decay rate if for all pinning $\eta$ on $\Lambda\subseteq V$, for any vertex $v \in V \setminus \Lambda$ and feasible spin $c \in \QQ_v$, and for any subset $W \subseteq V \setminus \Lambda \setminus \{v\}$ and two feasible configurations $\tau,\xi$ on $W$, it holds
\[
\left| \frac{\mu^\eta_v(c \mid \tau)}{\mu^\eta_v(c \mid \xi)} - 1 \right| \le C(1-\delta)^{\dist_G(v,W)}.
\]
\end{definition}

Our main result in this section is stated as follows.
We state it only for the hardcore model but it extends naturally to other spin systems as long as Glauber dynamics is ergodic under any pinning.

\begin{theorem}\label{thm:local-treewidth}
Let $G=(V,E)$ be an $n$-vertex graph of maximum degree $\Delta \ge 3$.
Suppose that $G$ has bounded local treewidth with constant parameters $a,d>0$, and suppose that the hardcore model on $G$ with fugacity $\lambda > 0$ satisfies SSM with constant parameters $C,\delta>0$. 
Then the mixing time of the Glauber dynamics for the hardcore Gibbs distribution on $G$ is $O(n \log^4n)$. 
\end{theorem}

%Let $G=(V,E)$ be a graph.
%For any subset $W\subseteq V$ of vertices, the (strong) diameter of $W$, denoted by $\diam(G[W])$, is the largest distance between two vertices from $W$ \emph{in the induced subgraph $G[W]$}. 
%In particular, if $G[W]$ is disconnected then $\diam_G(W) = \infty$. 

%Recall that for $S \subseteq U \subseteq V$ and $r \in \N$ we define
%\[
%\ball(S,r) = \{v \in V: \dist_G(v,S) \le r\}
%\]
%and
%\[
%\ball_U(S,r) = \{v \in U: \dist_G(v,S) \le r\}
%\]
%to be the portion of the ball around $S$ of radius $r$ inside $U$.

\subsection{Proof of \texorpdfstring{\cref{thm:local-treewidth}}{Theorem 5.3}}

For graphs of bounded local treewidth, if the diameter is mildly bounded (say poly-logarithmically), then the treewidth is also bounded. 
However, in general the diameter of $G$ can be arbitrarily large.  
We need the following low-diameter decomposition of graphs which allows us to focus on subgraphs of small diameters and thus of small treewidth. 
We remark that \cref{lem:decomp} holds for arbitrary graphs with no restriction on the local treewidth or maximum degree.

\begin{lemma}\label{lem:decomp}
Let $G = (V,E)$ be an $n$-vertex graph where $n \ge 10$. 
For any integer $r \in \N^+$, there exists a partition $V = \bigcup_{i=1}^m V_i$ of the vertex set satisfying the following conditions:
\begin{enumerate}
\item For each $i \in [m]$, we have $\diam(G[\ball(V_i,r)]) \le 6r\log n + 2r$;

\item For each vertex $v \in V$, we have $\left| \left\{ i \in [m]: v \in \ball(V_i,r) \right\} \right| \le 2\log n$.
%\[
%\left| \left\{ i \in [m]: v \in \ball(V_i,r) \right\} \right| \le 2\log n.
%\]
\end{enumerate}
\end{lemma}

The proof of \cref{lem:decomp} is postponed to \cref{subsec:decomp}, which is based on a classical low-diameter decomposition result of Linial and Saks \cite{LS93}.

To obtain block factorization of entropy from \cref{lem:decomp}, we need the strong spatial mixing property, which is given by the following lemma. 

\begin{lemma}\label{lem:SSM-EF}
Suppose SSM holds with constant parameters $C>0$ and $\delta \in (0,1)$. 
There exists a constant $\rho > 0$ such that for any subsets $S \subseteq U \subseteq V$ and any $\gamma \ge 10$, 
if $r \ge \rho (\log|S| + \log \gamma)$ then it holds for every function $f: \XX \to \R_{\ge 0}$ that,
\begin{equation}\label{eq:SSM-EF}
\E[\Ent_U f] \le e^{1/\gamma} \left( \E\left[\Ent_{\ball_U(S,r)} f\right] + \E\left[\Ent_{U \setminus S} f\right] \right).
\end{equation}
\end{lemma}

\begin{proof}
First observe that it suffices to establish the inequality under an arbitrary pinning $\eta$ on $V \setminus U$ and then \cref{eq:SSM-EF} follows by taking expectation over $\eta$. 
By \cref{lem:ent_fact_mar}, it then suffices to prove block factorization of entropy for the marginal distribution $\mu^\eta_{S\cup T}$ for the two blocks $S$ and $T = U \setminus \ball_U(S,r)$. 
Suppose $S = \{v_1,\dots,v_\ell\}$ where $\ell = |S|$ and let $\sigma$ be any feasible configuration on $S$.
For each $i$ define $S_i=\{v_1,\dots,v_i\}$ and let $\sigma_{S_i}$ be the configuration restricted on $S_i$. 
%Now, for sufficiently large $n$, 
By SSM we have that for any two feasible configurations $\tau,\xi$ on $T$,
\[
\frac{\mu^\eta_S(\sigma \mid \tau)}{\mu^\eta_S(\sigma \mid \xi)} 
= \prod_{i=1}^\ell \frac{\mu^{\eta,\sigma_{S_{i-1}}}_{v_i}(\sigma_{v_i} \mid \tau)}{\mu^{\eta,\sigma_{S_{i-1}}}_{v_i}(\sigma_{v_i} \mid \xi)} 
\le 
\left( 1+C e^{-\delta r} \right)^\ell
\le 
\left( 1+\frac{1}{4\gamma\ell} \right)^\ell
\le 1+\frac{1}{2\gamma},
\]
where we pick $r=\rho (\log \ell + \log \gamma)$ for large enough constant $\rho>0$ and we assume $\gamma \ge 10$ so that the last two inequalities hold.
Similarly, we also have 
\[
\frac{\mu^\eta_S(\sigma \mid \tau)}{\mu^\eta_S(\sigma \mid \xi)} \ge 1-\frac{1}{2\gamma}.
\]
Then, we deduce from \cref{lem:AT-weak-correlation} that the marginal distribution $\mu^\eta_{S\cup T}$ satisfies $\{S,T\}$-factorization of entropy with constant $C\le 1 + 1/\gamma \le e^{1/\gamma}$ (again we assume $\gamma \ge 10$ so that the first inequality holds).
%\[
%C \le 1+ \Theta(\frac{1}{n^3}) \le e^{1/n^2}.
%\]
\cref{eq:SSM-EF} then follows from \cref{lem:ent_fact_mar} and averaging over $\eta$.
\end{proof}

We now present the proof of \cref{thm:local-treewidth}.

\begin{proof}[Proof of \cref{thm:local-treewidth}]
By \cref{lem:SSM-EF}, there exists a constant $\rho>0$ such that \cref{eq:SSM-EF} holds for all subsets $S\subseteq U\subseteq V$ and for $\gamma = n$ whenever $r \ge \rho\log n$.
We apply \cref{lem:decomp} to $G$ for $r = \ceil{\rho\log n}$, and suppose the resulting partition is $V = \bigcup_{i=1}^m V_i$. 
For each $i \in [m]$ let $U_i = \bigcup_{j=i}^m V_j$, and so $U_1 = V$, $U_m = V_m$ and $U_{i+1} = U_i \setminus V_i$. 
Then we deduce from \cref{eq:SSM-EF} that for every function $f: \XX \to \R_{\ge 0}$,
\begin{align}
\Ent f &\le e^{1/n} \left( \E\left[\Ent_{\ball_{U_1}(V_1,r)} f\right] + \E\left[\Ent_{U_2} f\right] \right) \nonumber\\
&\le e^{2/n} \left( \E\left[\Ent_{\ball_{U_1}(V_1,r)} f\right] + \E\left[\Ent_{\ball_{U_2}(V_2,r)} f\right] + \E\left[\Ent_{U_3} f\right] \right) \nonumber\\
&\hspace{0.5em}\vdots \nonumber\\
&\le e^{m/n} \sum_{i=1}^m \E\left[\Ent_{\ball_{U_i}(V_i,r)} f\right] \nonumber\\
&\le 3 \sum_{i=1}^m \E\left[\Ent_{\ball(V_i,r)} f\right], \label{eq:decomp}
\end{align}
where the last inequality is due to $m \le n$ and that $\E[\Ent_B f] \le \E[\Ent_A f]$ for any $B \subseteq A \subseteq V$ (\cref{lem:MSW04}).
The good news are that each ball $\ball(V_i,r)$ has $O(\log^2n)$ diameter and thus $O(\log^{2d}n)$ treewidth, and every vertex is contained in $O(\log n)$ many balls; both properties are guaranteed by \cref{lem:decomp}.
Thus, to prove approximate tensorization of entropy for $\mu$, it suffices to prove it for each ball $\ball(V_i,r)$. 

We apply \cref{thm:strong-AT-sep-decomp} to show AT for each $B_i = \ball(V_i,r)$ (to be more precise, by \cref{thm:strong-AT-sep-decomp} we get AT uniformly under any pinning $\eta$ outside $B_i$ and then we take expectation over $\eta$). 
The balanced separator decomposition tree is given by \cref{lem:treewidth}.
Observe that the size of each separator is $O(\log^{2d}n)$ since the treewidth of $B_i$ is $O(\log^{2d}n)$, and the height $h$ of the decomposition tree satisfies $h = O(\log |B_i|) = O(\log n)$ by \cref{lem:height} since all separators are balanced. 
We take $r = \ceil{c \log\log n}$ in \cref{thm:strong-AT-sep-decomp} for some sufficiently large constant $c > 0$, such that block factorization for decomposition \cref{eq:st-decomp-factor} holds with $C_{U,S} = e^{1/h}$ for each node $(U,S)$; this again follows from \cref{lem:SSM-EF} where we have $|S| = O(\log^{2d}n)$ and $\gamma = h = O(\log n)$, and \cref{eq:SSM-EF} holds whenever $r \ge c \log\log n$.
Also $A\le 4\log n$ in \cref{eq:st-a} by \cref{lem:A}.
Therefore, we conclude from \cref{thm:strong-AT-sep-decomp} that every $B_i = \ball(V_i,r)$ satisfies approximate tensorization of entropy with multiplier
\begin{equation}\label{eq:C(B)}
C(\ball(V_i,r)) = 4\log n \cdot \big( e^{1/h} \big)^{h} \cdot \max_S C_S
\le 12\log n \cdot \max_S C_S,
\end{equation}
where $C_S$ is the AT multiplier in \cref{eq:st-AT-separa} for the ball $\ball_U(S,r)$ and we take maximum over all separators. 

Observe that, for each node $(U,S)$ in the decomposition tree we have
\[
|B_U(S,r)| = O\left(|S| \cdot \Delta^r \right) = O\left( \log^{2d}n \cdot \Delta^{c \log\log n} \right) \le \log^t n,
\]
for some constant $t> 0$ when $n$ is sufficiently large. 
Define $\varphi(k)$ to be the maximum of optimal AT multipliers over all subsets of vertices of size at most $k$; namely, for all $W \subseteq V$ with $|W| \le k$, it holds for every function $f: \XX \to \R_{\ge 0}$ that 
\[
\E[\Ent_W f] \le \varphi(k) \sum_{v \in W} \E[\Ent_v f].
\]
Note that $\varphi(k)$ is monotone increasing.
Thus, we obtain from \cref{eq:C(B)} that
\begin{equation}\label{eq:g-C(B)}
C(\ball(V_i,r)) \le 12\log n \cdot \varphi(\log^t n).
\end{equation}
Combining \cref{eq:decomp,eq:g-C(B)} and \cref{lem:decomp}, we obtain
\begin{align*}
\Ent f &\le 3 \sum_{i=1}^m \E\left[\Ent_{\ball(V_i,r)} f\right] \\
&\le 3 \cdot 12\log n \cdot \varphi(\log^t n) \sum_{i=1}^m \sum_{v \in \ball(V_i,r)} \E\left[\Ent_v f\right] \\
&\le 100 \log^2n \cdot \varphi(\log^t n) \sum_{v \in V} \E\left[\Ent_v f\right].
\end{align*}
More generally, for any subset $W \subseteq V$ of size $k_0 \le |W| \le k$ where $k_0>0$ is some fixed constant, we have for every function $f: \XX \to \R_{\ge 0}$ that 
\[
\E[\Ent_W f] 
\le 100 \log^2k \cdot \varphi(\log^t k) \sum_{v \in W} \E\left[\Ent_v f\right].
\]
This is shown by the same arguments applied to $\mu^\eta_W$ under an arbitrary pinning $\eta$ outside $W$ and then taking expectation.
Hence, we have established the following recursive bound: for all $k \ge k_0$, 
%By applying the whole arguments to any subset of vertices, we have shown that for all $n \ge k_0$ where $k_0>0$ is some constant, it holds
%\[
%\varphi(n) \le 50 \log^2n \cdot \varphi(\log^t n).
%\] 
%In fact, the same arguments have established the recursive bound
\begin{equation}\label{eq:recursive-phi}
\varphi(k) \le \max\left\{ 100 \log^2k \cdot \varphi(\log^t k), \, \varphi(k_0) \right\}.
\end{equation}

We now solve \cref{eq:recursive-phi}.
By choosing $k_0$ large enough, we may assume that for all $k \ge k_0$, 
\[
\log^t k < k
\quad \text{and} \quad
100 t^3 (\log \log k)^3 \le \log k.
\]
We prove by induction that 
\begin{equation}\label{eq:induction}
\varphi(k) \le \varphi(k_0) \cdot \log^3k.
\end{equation}
\cref{eq:induction} is trivial for $k < k_0$. 
Suppose \cref{eq:induction} is true for all $k' < k$. 
Then we have
%Then from the recursive bound \cref{eq:recursive-phi} we have
\[
%\varphi(k) %\le 
100 \log^2k \cdot \varphi(\log^t k) 
\le 100 \log^2k \cdot \varphi(k_0) \cdot t^3 (\log \log k)^3
\le \varphi(k_0) \cdot \log^3k.
\]
Thus, we deduce from the recursive bound \cref{eq:recursive-phi} that 
$$ \varphi(k) \le \max\left\{ \varphi(k_0) \cdot \log^3k, \, \varphi(k_0) \right\} = \varphi(k_0) \cdot \log^3k, $$ 
establishing \cref{eq:induction}.

Therefore, we have $\varphi(n) = O(\log^3n)$.
In particular, the hardcore Gibbs distribution on $G$ satisfies AT with multiplier $O(\log^3n)$, and the mixing time follows from \cref{lem:AT-Glauber}.
\end{proof}

\subsection{Rapid mixing via SSM for graphs of bounded growth}
\label{subsec:growth}

In this subsection we consider graphs with polynomially bounded neighborhood growth.

\begin{definition}[Bounded Growth]
Let $G=(V,E)$ be a graph and $a,d > 0$ be reals. 
We say $G$ has polynomially bounded growth if for any vertex $v$ and any integer $r \ge 1$ it holds
\[
|\ball(v,r)| \le a \cdot r^d. 
\]
\end{definition}
Observe that any family of graphs of bounded growth also have bounded maximum degree and bounded local treewidth. 
To see the latter, for any nontrivial subgraph $H$ of a graph $G$ of bounded growth, let $v$ be any vertex of $H$ and $r=\diam(H) \ge 1$, and we have
\begin{align}\label{eq:twH}
\tw(H) \le |H| \le |\ball(v,r)| \le a \cdot r^d,
\end{align}
showing that $G$ has bounded local tree width.
Thus, we deduce from \cref{thm:local-treewidth} that for graphs of bounded growth, SSM implies $O(n\log^4n)$ mixing of Glauber dynamics. 
We can actually improve this mixing time bound slightly.

\begin{theorem}\label{thm:growth}
Let $G=(V,E)$ be an $n$-vertex graph.
Suppose that $G$ has bounded growth with constant parameters $a,d>0$, and suppose that the hardcore model on $G$ with fugacity $\lambda > 0$ satisfies SSM with constant parameters $C,\delta>0$. 
Then the mixing time of the Glauber dynamics for the hardcore Gibbs distribution on $G$ is $O(n \log^3n)$. 
\end{theorem}

\begin{proof}
Following the proof of \cref{thm:local-treewidth}, we have that for every function $f: \XX \to \R_{\ge 0}$,
\[
\Ent f
\le 3 \sum_{i=1}^m \E\left[\Ent_{\ball(V_i,r)} f\right],
\]
which is \cref{eq:decomp}.
For bounded-growth graphs, we have similarly as \cref{eq:twH} that
\[
|\ball(V_i,r)| \le a \cdot \big( \diam(G[\ball(V_i,r)]) \big)^d \le t\log^{2d} n,
\]
where $t > 0$ is some fixed constant, and the last inequality is due to \cref{lem:decomp}.
With the same definition of $\varphi(\cdot)$, we deduce that
\begin{align*}
\Ent f &\le 3 \sum_{i=1}^m \E\left[\Ent_{\ball(V_i,r)} f\right] \\
&\le 3 \cdot \varphi(t\log^{2d} n) \sum_{i=1}^m \sum_{v \in \ball(V_i,r)} \E\left[\Ent_v f\right] \\
&\le 6 \log n \cdot \varphi(t\log^{2d} n) \sum_{v \in V} \E\left[\Ent_v f\right],
\end{align*}
where the last inequality follows from \cref{lem:decomp}.
This allows us to obtain the following recursive bound: 
for all $k \ge k_0$ where $k_0 > 0$ is some fixed constant, 
\begin{equation}\label{eq:2recursive-phi}
\varphi(k) \le \max\left\{ 6 \log k \cdot \varphi(t\log^{2d} k), \, \varphi(k_0) \right\}.
\end{equation}
Solving \cref{eq:2recursive-phi}, we can get $\varphi(n) = O(\log^2n)$.
Thus, AT of entropy holds with multiplier $O(\log^2n)$, and we get the mixing time bound from \cref{lem:AT-Glauber}.
\end{proof}

\subsection{Low-diameter decomposition}
\label{subsec:decomp}

In this subsection we present the proof of \cref{lem:decomp}.

Given a graph $G=(V,E)$ and a partition $V = \bigcup_{i=1}^m V_i$ of the vertex set into $m$ clusters, define the quotient graph $\HH(G;V_1,\dots,V_m)$ to be the graph with vertex set $\{V_i: i \in [m]\}$ where two clusters $V_i,V_j$ are adjacent iff there exists $u\in V_i,v\in V_j$ such that $uv \in E$. Namely, $H$ is the graph obtained from $G$ by contracting every cluster into a single vertex and connect two vertices iff the two clusters are adjacent. 
We need the following classical result of low-diameter decomposition due to Linial and Saks \cite{LS93}.

\begin{lemma}[{\cite[Theorem 2.1]{LS93}}]
\label{lem:Linial-Saks}
Let $G = (V,E)$ be an $n$-vertex graph where $n \ge 10$. 
There exists a partition $V = \bigcup_{i=1}^m V_i$ of the vertex set into clusters such that the following conditions hold:
\begin{enumerate}
\item For each $i \in [m]$, we have $\diam(G[V_i]) \le 3\log n$;

\item The quotient graph $\HH = \HH(G;V_1,\dots,V_m)$ has chromatic number at most $2\log n$.
\end{enumerate}
\end{lemma}

\begin{remark}
\cref{lem:Linial-Saks} is obtained by taking $p=1/2$ in Theorem 2.1 of \cite{LS93}. Also note that the partition in \cref{lem:Linial-Saks} is defined differently from \cite{LS93}: the cluster in \cite{LS93} represents the union of all clusters of the same color in \cref{lem:Linial-Saks}. 
\end{remark}

\begin{proof}[Proof of \cref{lem:decomp}]
Let $G^{\le 2r}$ be the graph with vertex set $V$ where two vertices are adjacent iff their graph distance is at most $2r$.
By \cref{lem:Linial-Saks}, there exists a partition $V = \bigcup_{i=1}^m V_i$ for the graph $G^{\le 2r}$ such that
\begin{enumerate}
\item For each $i \in [m]$, we have $\diam(G^{\le 2r}[V_i]) \le 3\log n$;

\item The quotient graph $\HH(G^{\le 2r};V_1,\dots,V_m)$ has chromatic number at most $2\log n$.
\end{enumerate}
We claim that such a partition $V = \bigcup_{i=1}^m V_i$ satisfies our requirements. 

Fix $i \in [m]$.
For $u,v \in V_i$, since the diameter of $G^{\le 2r}[V_i]$ is at most $3\log n$, 
there exists a path $P$ in $G^{\le 2r}[V_i]$ connecting $u$ and $v$ of length at most $3\log n$. 
By replacing every edge in $P$ in $G^{\le 2r}[V_i]$ with a path in $G$ of length at most $2r$, we obtain a path $P'$ in $G$ of length at most $3\log n \cdot 2r = 6r\log n$. 
In particular, this new path $P'$ is contained in the induced subgraph $G[\ball(V_i,r)]$.
Hence, the distance between all pairs of $u,v \in V_i$ in $G[\ball(V_i,r)]$ is at most $6r\log n$.
Also, note that every vertex in $\ball(V_i,r)$ is at distance at most $r$ from $V_i$ in $G[\ball(V_i,r)]$.
Thus, the diameter of $G[\ball(V_i,r)]$ is at most $6r\log n + 2r$, verifying the first condition.

For the second condition, take an arbitrary vertex $v \in V$ and consider all clusters at distance at most $r$ from $v$. Then all these clusters have pairwise distance at most $2r$, meaning that they form a clique in the quotient graph $\HH(G^{\le 2r};V_1,\dots,V_m)$. 
Since the quotient graph is $(2\log n)$-colorable, the size of this clique is at most $2\log n$, implying the second condition.
\end{proof}

\section{Proofs for Variance and Entropy Factorization}
\label{sec:proofs}

In this section, we give missing proofs from \cref{subsec:tools,subsec:marginal}.

\subsection{Proof of \texorpdfstring{\cref{lem:AT-weak-correlation}}{Lemma 3.3}}
\label{subsec:AT-weak-correlation}

In this subsection we present the proof of \cref{lem:AT-weak-correlation}. 
Our proof here avoids some technical difficulties appeared in \cite[Proposition 2.1]{Cesi01} or \cite[Lemma 5.2]{DPP02}, and allows us to get a slightly better constant for AT.

\begin{proof}[Proof of \cref{lem:AT-weak-correlation}]
We notice that factorization of entropy \cref{eq:AT-wc-goal} implies factorization of variance \cref{eq:var-AT-wc-goal} with the same constant by a standard linearization argument (plugging $f = 1+\theta g$ into \cref{eq:AT-wc-goal} and then taking $\theta \to 0$), see \cite{CMT15}. 
Thus, it suffices to consider only entropy and prove \cref{eq:AT-wc-goal}. 

By the law of total entropy (\cref{lem:CMT}) we have
\[
\Ent f = \E[\Ent_X f] + \Ent[\E_X f] = \E[\Ent_Y f] + \Ent[\E_Y f].
\]
Hence, \cref{eq:AT-wc-goal} is equivalent to that
\begin{equation}\label{eq:AT-wc-2goal}
\Ent f \ge (1-\eps) \left( \Ent[\E_X f] + \Ent[\E_Y f] \right).
\end{equation}
It suffices to show \cref{eq:AT-wc-2goal}.

Without loss of generality we may assume that $\E f = 1$. Thus, $\E[\E_X f] = \E[\E_Y f] = \E f = 1$.
Let us define
\begin{equation}\label{eq:D-def}
D = \Ent f - \Ent[\E_X f] - \Ent[\E_Y f].
\end{equation}
Then by definition we have
\begin{align*}
D 
&= \E[f \log f] - \E[ (\E_X f) \log (\E_X f) ] - \E[ (\E_Y f) \log (\E_Y f) ] \\
&= \E[f \log f] - \E[ f \log (\E_X f) ] - \E[ f \log (\E_Y f) ] \\
&= \E\left[ f \log f - f \log \left( (\E_X f)(\E_Y f) \right) \right].
\end{align*}
Let $(x,y) \in \XX \times \YY$ such that $\pi(x,y) > 0$. 
%Suppose $(\E_X^y f) (\E_Y^x f) > 0$.
%Using the fact that $a \log(a/b) \ge a-b$ for all $a\ge 0$ and $b > 0$, we get that at the point $(x,y)$
%\begin{align*}
%f \log f - f \log (\E_X f) - f \log (\E_Y f)
%= f \log\left( \frac{f}{(\E_X f) (\E_Y f)} \right)
%\ge f - (\E_X f) (\E_Y f).
%\end{align*}
%Meanwhile, if $\E_X^y f = 0$ then $f(x,y) = 0$ since $f$ is non-negative, and hence at the point $(x,y)$ we have 
%\[
%f \log f - f \log (\E_X f) - f \log (\E_Y f) = 0 = f - (\E_X f) (\E_Y f),
%\] 
%by the convention $0\log 0 = 0$. 
%The case $\E_Y^x f = 0$ is similar. 
If $(\E_X^y f) (\E_Y^x f) > 0$,
then by the inequality $a \log(a/b) \ge a-b$ for all $a\ge 0$ and $b > 0$, we deduce that at the point $(x,y)$
\begin{align}\label{eq:simplify}
f \log f - f \log \left( (\E_X f)(\E_Y f) \right)
\ge f - (\E_X f) (\E_Y f).
\end{align}
Meanwhile, if $(\E_X^y f) (\E_Y^x f) = 0$ then it must hold $f(x,y) = 0$ since $f$ is non-negative, 
and hence \cref{eq:simplify} still holds at $(x,y)$ with both sides equal to 0 (recall $0\log 0 = 0$). 
Thus, we obtain that
\begin{equation}\label{eq:two-step1}
D
\ge \E\left[ f - (\E_X f) (\E_Y f) \right] 
%= 1 - \E\left[ (\E_X f) (\E_Y f) \right] 
= - \E\left[ (\E_X f -1) (\E_Y f -1) \right],
\end{equation}
where we use $\E[\E_X f] = \E[\E_Y f] = \E f = 1$.
Note that $\E\left[ (\E_X f -1) (\E_Y f -1) \right]$ is the covariance of the two functions $\E_X f$ and $\E_Y f$.

Let us now define a probability distribution over $\XX \times \YY$ by $\nu = f \pi$, i.e., $\nu(x,y) = \pi(x,y) f(x,y)$ for all $(x,y) \in \XX \times \YY$. Note that $\nu$ is indeed a distribution since $\E_\pi f = 1$. 
We also define the marginal distributions $\nu_X, \nu_Y$ and the conditional distributions $\nu_X^y,\nu_Y^x$ similarly as for $\pi$. 
Observe that for each $y \in \YY$, we have
\[
\E_X^y f = \sum_{x \in \XX} \pi_X^y f(x,y) = \sum_{x \in \XX} \frac{\pi(x,y)}{\pi_Y(y)} f(x,y) = \frac{1}{\pi_Y(y)} \sum_{x \in \XX} \nu(x,y)
= \frac{\nu_Y(y)}{\pi_Y(y)}.
\]
Then, we deduce that
\begin{equation}\label{eq:two-plug}
\E\left[ (\E_X f -1) (\E_Y f -1) \right]
= \sum_{(x,y) \in \XX \times \YY} \pi(x,y) \left( \frac{\nu_X(x)}{\pi_X(x)} -1 \right) \left( \frac{\nu_Y(y)}{\pi_Y(y)} -1 \right).
\end{equation}

Let $\XX^+ = \{x \in \XX: \nu_X(x) \ge \pi_X(x)\}$ and $\XX^- = \{x \in \XX: \nu_X(x) < \pi_X(x)\}$, so $(\XX^+, \XX^-)$ is a partition of $\XX$. Define $\YY^+$ and $\YY^-$ in the same way with respect to $Y$. 
Recall that the condition in \cref{eq:AT-wc-cond} is equivalent to that for all $(x,y) \in \XX \times \YY$,
\[
(1-\eps) \pi_X(x) \pi_Y(y) \le \pi(x,y) \le (1+\eps) \pi_X(x) \pi_Y(y).
\]
Hence, we obtain that
\begin{align*}
&\sum_{(x,y) \in \XX^+ \times \YY^+} \pi(x,y) \left( \frac{\nu_X(x)}{\pi_X(x)} -1 \right) \left( \frac{\nu_Y(y)}{\pi_Y(y)} -1 \right) \\
\le{}& (1+\eps) \sum_{(x,y) \in \XX^+ \times \YY^+} \pi_X(x) \pi_Y(y) \left( \frac{\nu_X(x)}{\pi_X(x)} -1 \right) \left( \frac{\nu_Y(y)}{\pi_Y(y)} -1 \right) \\
={}& (1+\eps) \left( \sum_{x \in \XX^+} \nu_X(x) - \pi_X(x) \right) \left( \sum_{y \in \YY^+} \nu_Y(y) - \pi_Y(y) \right) \\
={}& (1+\eps) \dtv(\nu_X, \pi_X) \dtv(\nu_Y, \pi_Y). 
\end{align*}
The same upper bound holds for $\XX^- \times \YY^-$ as well. 
Meanwhile,
\begin{align*}
&\sum_{(x,y) \in \XX^+ \times \YY^-} \pi(x,y) \left( \frac{\nu_X(x)}{\pi_X(x)} -1 \right) \left( \frac{\nu_Y(y)}{\pi_Y(y)} -1 \right) \\
\le{}& (1-\eps) \sum_{(x,y) \in \XX^+ \times \YY^-} \pi_X(x) \pi_Y(y) \left( \frac{\nu_X(x)}{\pi_X(x)} -1 \right) \left( \frac{\nu_Y(y)}{\pi_Y(y)} -1 \right) \\
={}& (1-\eps) \left( \sum_{x \in \XX^+} \nu_X(x) - \pi_X(x) \right) \left( \sum_{y \in \YY^-} \nu_Y(y) - \pi_Y(y) \right) \\
={}& -(1-\eps) \dtv(\nu_X, \pi_X) \dtv(\nu_Y, \pi_Y), 
\end{align*}
and the same bound also holds for $(x,y) \in \XX^- \times \YY^+$. 
Therefore, plugging into \cref{eq:two-plug}, we obtain that
\begin{align}
\E\left[ (\E_X f -1) (\E_Y f -1) \right] 
&\le 2(1+\eps) \dtv(\nu_X, \pi_X) \dtv(\nu_Y, \pi_Y) - 2(1-\eps) \dtv(\nu_X, \pi_X) \dtv(\nu_Y, \pi_Y) \nonumber\\
&= 4\eps \dtv(\nu_X, \pi_X) \dtv(\nu_Y, \pi_Y). \label{eq:two-step2}
\end{align}
By Pinsker's inequality, we have
\begin{align}
4\dtv(\nu_X, \pi_X) \dtv(\nu_Y, \pi_Y) 
&\le 2 \sqrt{\DKL(\nu_X \Vert \pi_X) \DKL(\nu_Y \Vert \pi_Y)} \nonumber\\
&\le \DKL(\nu_X \Vert \pi_X) + \DKL(\nu_Y \Vert \pi_Y) \nonumber\\
&= \Ent[\E_Y f] + \Ent[\E_X f] \label{eq:two-step3}
\end{align}
where the last equality follows from the observation that $\DKL(\nu_X \Vert \pi_X) = \Ent[\nu_X / \pi_X] = \Ent[\E_Y f]$
and similarly $\DKL(\nu_Y \Vert \pi_Y) = \Ent[\E_X f]$. 

Finally, combining \cref{eq:two-step1,eq:two-step2,eq:two-step3}, we obtain that
\[
D \ge - \eps \left( \Ent[\E_X f] + \Ent[\E_Y f] \right).
\]
Recalling \cref{eq:D-def}, we obtain \cref{eq:AT-wc-2goal} as claimed and the lemma then follows.
\end{proof}

\subsection{Proofs of \texorpdfstring{\cref{lem:AT-strong-correlation,lem:AT-crude-bound}}{Lemmas 3.4 and 3.6}}
\label{subsec:crude-bound}

Before presenting the proofs of these two lemmas, we first give some definitions and lemmas that are needed.

Our proof of approximate tensorization is based on the spectral independence approach. 
The following definitions and theorem are taken from \cite{FGYZ21} which builds upon \cite{AL20,ALO20} (see also \cite{CGSV21}). 

\begin{definition}[Influence Matrix, \cite{FGYZ21}]
\label{def:inf-matrix}
Let $\pi$ be a distribution over a finite product space $\XX = \prod_{i=1}^n \XX_i$. 
Fix a subset $\Lambda \subseteq [n]$ and a feasible partial assignment $x_\Lambda \in \prod_{i \in \Lambda} \XX_i$ with $\pi_\Lambda(x_\Lambda) > 0$. 
For any $i,j \in [n] \setminus \Lambda$ with $i\neq j$, we define the (pairwise) influence of $X_i$ on $X_j$ conditioned on $x_\Lambda$ by
\[
\Psi_\pi^{x_\Lambda} (i,j) = \max_{x_i,x'_i} \dtv( \pi_j^{x_\Lambda,x_i}, \pi_j^{x_\Lambda,x'_i} ), 
\]
where $x_i,x'_i$ are chosen from $\XX_i$ such that $\pi_i^{x_\Lambda}(x_i) > 0$ and $\pi_i^{x_\Lambda}(x'_i) > 0$. 

The (pairwise) influence matrix $\Psi_\pi^{x_\Lambda}$ is defined with entries given as above and also with $\Psi_\pi^{x_\Lambda} (i,i) = 0$ for $i \in [n] \setminus \Lambda$ for diagonal entries.  
\end{definition}

\begin{definition}[Spectral Independence, \cite{FGYZ21}]
We say a distribution $\pi$ over a finite product space $\XX = \prod_{i=1}^n \XX_i$ is $(\eta_0,\eta_1,\dots,\eta_{n-2})$-spectrally independent, if for every $0 \le k \le n-2$, every subset $\Lambda \subseteq [n]$ of size $|\Lambda| = k$, and every feasible partial assignment $x_\Lambda \in \prod_{i \in \Lambda} \XX_i$ with $\pi_\Lambda(x_\Lambda) > 0$, the spectral radius $\rho(\Psi_\pi^{x_\Lambda})$ of the influence matrix $\Psi_\pi^{x_\Lambda}$ satisfies
\[
\rho(\Psi_\pi^{x_\Lambda}) \le \eta_k.
\]
\end{definition}

\begin{theorem}[{\cite[Theorem 3.2]{FGYZ21}}]
\label{thm:FGYZ-SI}
Let $\pi$ be a distribution over a finite product space $\XX = \prod_{i=1}^n \XX_i$. 
Let $\eta_0,\eta_1,\dots,\eta_{n-2}$ be a sequence of reals such that $0\le \eta_k < n-k-1$ for each $k$. 
If $\pi$ is $(\eta_0,\eta_1,\dots,\eta_{n-2})$-spectrally independent, then the spectral gap of the Glauber dynamics $P$ for sampling from $\pi$ satisfies
\[
\lambda(P) \ge \frac{1}{n} \prod_{k=0}^{n-2} \left( 1 - \frac{\eta_k}{n-k-1} \right). 
\]
\end{theorem}

We also need the following well-known comparison between the spectral gap and the standard log-Sobolev constant (which holds more generally for any Markov chain). 

\begin{lemma}[{\cite[Corollay A.4]{DS96}}]
\label{lem:LSC-gap}
Let $\pi$ be a distribution over a finite product space $\XX = \prod_{i=1}^n \XX_i$. 
%Let $P$ denote the Glauber dynamics for $\pi$. 
Suppose $\lambda$ is the spectral gap of the Glauber dynamics for sampling from $\pi$, and $\rho$ is the standard log-Sobolev constant of it. 
Then we have
\[
\rho \ge \frac{1-2\pi_{\min}}{\log(1/\pi_{\min} - 1)} \lambda,
%\ge \frac{1}{2+\log(1/\pi_{\min})} \lambda(P),
\]   
where $\pi_{\min} = \min_{x \in \XX:\, \pi(x) > 0} \pi(x)$. 
In particular, if $\pi$ has positive density on at least two assignments in $\XX$ then $\pi_{\min} \le 1/2$ and we have
\[
\rho \ge \frac{\lambda}{2+\log(1/\pi_{\min})}. 
\] 
\end{lemma}

Finally, we need the following lemma for deriving AT from the spectral gap and the standard log-Sobolev constant. 

\begin{lemma}[{\cite[Proposition 1.1]{CMT15}}]
\label{lem:CMT}
Let $\pi$ be a distribution over a finite product space $\XX = \prod_{i=1}^n \XX_i$. 
%Let $P$ denote the Glauber dynamics for $\pi$. 
Suppose $\lambda$ is the spectral gap of the Glauber dynamics for sampling from $\pi$, and $\rho$ is the standard log-Sobolev constant of it. 
Then we have
%for every function $f: \XX \to \R_{\ge 0}$, it holds
\begin{align*}
\Var f &\le \frac{1}{\lambda n} \sum_{i=1}^n \E[\Var_i f], \quad \forall f: \XX \to \R \\
\text{and} \quad
\Ent f &\le \frac{1}{\rho n} \sum_{i=1}^n \E[\Ent_i f], \quad \forall f: \XX \to \R_{\ge 0}.
\end{align*}
\end{lemma}

We are now ready to prove \cref{lem:AT-strong-correlation,lem:AT-crude-bound}.

\begin{proof}[Proof of \cref{lem:AT-strong-correlation}]
The spectral radius of the influence matrix of $\pi$, as defined in \cref{def:inf-matrix}, is upper bounded by
\[
\rho(\Psi_\pi) \le
\rho \left(
\begin{bmatrix}
0 & 1-\eps_Y\\
1-\eps_X & 0\\
\end{bmatrix}
\right)
= \sqrt{(1-\eps_X)(1-\eps_Y)}
\le 1-\frac{\eps_X + \eps_Y}{2}.
\]
Therefore, by \cref{thm:FGYZ-SI} the spectral gap of the Glauber dynamics is at least $(\eps_X+\eps_Y)/4$, 
and by \cref{lem:LSC-gap} the standard log-Sobolev constant is at least $(\eps_X+\eps_Y)/(8+4\log(1/\pi_{\min}))$.
%\[
%\frac{1-\sqrt{\eps_X \eps_Y}}{2(2+\log(1/\pi^*))}.
%\]
We then deduce the lemma from \cref{lem:CMT}.
%a comparison between the log-Sobolev constant and the AT constant from \cite{CMT15}.
\end{proof}

\begin{remark}
Another way of proving \cref{lem:AT-strong-correlation} is by viewing the (pairwise) influence matrix $\Psi_\pi$ as the Dobrushin dependency/influence matrix, and showing the Glauber dynamics is contractive with respective to some weighted Hamming distance with the weight vector given by the principle eigenvector of $\Psi_\pi$; the lower bound on the spectral gap then follows from \cite[Theorem 13.1]{LP17} or \cite{Chen98}.
\end{remark}

\begin{proof}[Proof of \cref{lem:AT-crude-bound}]
By definition we see that $\pi$ is $(\eta_0,\eta_1,\dots,\eta_{n-2})$-spectrally independent, 
where for each $k$ we have
\[
\eta_k \le (n-k-1) (1-\eps)
\]
by considering the $\ell_\infty$ norm (absolute row sum) of the influence matrices. 
Hence, we deduce from \cref{thm:FGYZ-SI} that the spectral gap of the Glauber dynamics is lower bounded by
\[
\lambda \ge \frac{1}{n} \prod_{k=0}^{n-2} \left( 1 - \frac{\eta_k}{n-k-1} \right) 
\ge \frac{\eps^{n-1}}{n}.
\]
Furthermore, by \cref{lem:LSC-gap} the standard log-Sobolev constant is lower bounded by
\[
\rho \ge \frac{\lambda}{2+\log(1/\pi_{\min})}
\ge \frac{\eps^{n-1}}{(2+\log(1/\pi_{\min}))n}.
\]   
%Then, it follows from \cref{lem:LSC-gap,lem:CMT} that $\pi$ satisfies AT with constant.
Finally, the lemma follows from an application of \cref{lem:CMT}.
\end{proof}

\subsection{Proof of \texorpdfstring{\cref{lem:ent_fact_mar}}{Lemma 4.10}}
\label{subsec:ent_fact_mar}

\begin{proof}[Proof of \cref{lem:ent_fact_mar}]
Recall, the marginal distribution $\pi_{XY}$ satisfies $\{\{X\}, \{Y\}\}$-factorization (i.e., approximate tensorization) of entropy with constant $C$ if 
\begin{equation}\label{eq:2}
\Ent_{XY} g \le C \left( \E_{XY}\left[ \Ent_X g \right] + \E_{XY}\left[ \Ent_Y g \right] \right), 
\quad \forall g: \XX \times \YY \to \R_{\ge 0}.
\end{equation}
We claim that \cref{eq:2} is equivalent to 
\begin{equation}\label{eq:1}
\Ent \bar{g} \le C \left( \E\left[ \Ent_{XZ} \bar{g} \right] + \E\left[ \Ent_{YZ} \bar{g} \right] \right), 
\quad \forall \bar{g}: \XX \times \YY \times \ZZ \to \R_{\ge 0} \text{~depending \emph{only} on $X$ and $Y$}
\end{equation}
where the underlying distribution is $\pi = \pi_{XYZ}$.
To see this, observe that every function $g: \XX \times \YY \to \R_{\ge 0}$ is in one-to-one correspondence to a function $\bar{g}: \XX \times \YY \times \ZZ \to \R_{\ge 0}$ depending only on $X,Y$ by the relationship $\bar{g}(x,y,z) = g(x,y)$. 
By definitions, we have $\Ent_{XY} g = \Ent \bar{g}$, $\Ent_{\pi^y_X} g = \Ent_{\pi^y_{XZ}} \bar{g}$ for all $y \in \YY$, and $\Ent_{\pi^x_Y} g = \Ent_{\pi^x_{YZ}} \bar{g}$ for all $x \in \XX$, implying \cref{eq:2,eq:1} are equivalent.

Therefore, it suffices to show that \cref{eq:1} is equivalent to that $\pi$ satisfies $\{\{X,Z\}, \{Y,Z\}\}$-factorization of entropy with constant $C$, i.e.,
\begin{equation}\label{eq:xx}
\Ent f \le C \left( \E\left[ \Ent_{XZ} f \right] + \E\left[ \Ent_{YZ} f \right] \right),
\quad \forall f: \XX \times \YY \times \ZZ \to \R_{\ge 0}.
\end{equation}
It is trivial that \cref{eq:xx} implies \cref{eq:1}. 
For the other direction, 
suppose \cref{eq:1} is true. 
Since $\E_Z f$ is a function depending only on $X$ and $Y$,  
we have from \cref{eq:1} that
\[
\Ent(\E_Z f) \le C \left( \E\left[ \Ent_{XZ} (\E_Z f) \right] + \E\left[ \Ent_{YZ} (\E_Z f) \right] \right). 
\]
Then, we deduce from \cref{lem:MSW04} that
\begin{align*}
\Ent f &= \E[\Ent_Z f] + \Ent(\E_Z f)\\ 
&\le \E[\Ent_Z f] + C \left( \E\left[ \Ent_{XZ} (\E_Z f) \right] + \E\left[ \Ent_{YZ} (\E_Z f) \right] \right)\\
&= \E[\Ent_Z f] + C \left( \E\left[ \Ent_{XZ} f \right] - \E\left[ \Ent_Z f \right] + \E\left[ \Ent_{YZ} f \right] - \E\left[ \Ent_Z f \right] \right) \\
&= C \left( \E\left[ \Ent_{XZ} f \right] + \E\left[ \Ent_{YZ} f \right] \right) - (2C-1) \E[\Ent_Z f] \\
&\le C \left( \E\left[ \Ent_{XZ} f \right] + \E\left[ \Ent_{YZ} f \right] \right),
\end{align*}
where the last inequality is due to $C \ge 1$ (this can be seen by considering functions depending only on $X$ in \cref{eq:1}). 
\end{proof}

%\section{Conclusion}

\bibliographystyle{alpha}
\bibliography{ent.bib}

\end{document}